\documentclass[12pt]{article}
\usepackage[utf8]{inputenc}
\usepackage{hyperref}

\usepackage{latexsym,amsmath,amssymb,amsthm,mathtools}
\usepackage{biblatex}

\newtheorem{theorem}{Theorem}
\usepackage{graphicx}
\graphicspath{ {./images/} } 
\newcommand{\ind}{\mathrel{\text{\scalebox{1.07}{$\perp\mkern-10mu\perp$}}}}

\usepackage{color,ulem}
\definecolor{dred}{rgb}{0.5,0,0}
\definecolor{dgreen}{rgb}{0,0.3,0}
\definecolor{dblue}{rgb}{0,0,0.5}
\definecolor{dmagenta}{cmyk}{0,1,0,0.6}
\definecolor{dcyan}{cmyk}{1,0,0,0.5}
\definecolor{grey}{gray}{0.9}
\definecolor{orange}{rgb}{1,0.65,0}
\definecolor{mellow}{rgb}{.847,.72,.525}
\definecolor{golden}{rgb}{.80392,.60784,.11373}
\definecolor{dgolden}{rgb}{.5451,.39608,.03137}
\definecolor{brown}{rgb}{.15,.15,.15}
\definecolor{darkolivegreen}{rgb}{.33333,.41961,.18431}

\newcommand{\COR}{{\mbox{COR}}}
\newcommand{\OR}{{\mbox{OR}}}
\newcommand{\pa}{\mbox{\rm pa}}
\newcommand{\nd}{\mbox{\rm nd}}
\newcommand{\pre}{{\mbox{\rm pre}}}

\newcommand{\PA}{\mbox{\rm PA}}
\newcommand{\ND}{\mbox{\rm ND}}
\newcommand{\PRE}{\mbox{\rm PRE}}
\newcommand{\nb}{\mbox{\rm nb}}

\usepackage{bm}

\def\bzero{{\bf 0}}

\def\bone{{\bf 1}}

\newcommand{\dD}{{\rm I\hspace{-.7mm}D}}

\newcommand{\cov}{{\mbox{\rm cov}}}

\newcommand{\step}{{\mbox{\rm \it step}}}

\def\bbeta{{\bm\beta}}

\def\blambda{{\bm\lambda}}
\def\bmu{{\bm\mu}}

\def\btau{{\bm\tau}}

\def\bLambda{{\bm\Lambda}}

\def\bmm{{\bf m}}
\def\bn{{\bf n}}

\def\bs{{\bf s}}

\def\bu{{\bf u}}

\def\by{{\bf y}}

\def\bB{{\bf B}}
\def\bC{{\bf C}}
\def\bD{{\bf D}}

\def\bI{{\bf I}}
\def\bJ{{\bf J}}

\def\bM{{\bf M}}

\def\bU{{\bf U}}
\def\bV{{\bf V}}

\def\bX{{\bf X}}

\def\cK{{\mathcal K}}

\def\cM{{\mathcal M}}

\def\cV{{\mathcal V}}

\def\cX{{\mathcal X}}

\begin{document}

\title{Marginal Models: an Overview}
\author{Tam\'as Rudas\\Department of Statistics, Faculty of Social
Sciences\\
E\"{o}tv\"{o}s Lor\'{a}nd University, Budapest\\
\texttt{trudas@elte.hu}
\and 
Wicher Bergsma\\
London School of Economics\\
\texttt{W.P.Bergsma@lse.ac.uk} }
\maketitle

\begin{abstract}
Marginal models involve restrictions on the conditional and marginal association structure of a set of categorical variables. They generalize log-linear models for contingency tables, which are the fundamental tools for modelling the conditional association structure.
This chapter gives an overview of the development of marginal models during the past 20 years.
After providing some motivating examples, the first few sections focus on the definition and characteristics of marginal models. Specifically, we show how their fundamental properties can be understood from the properties of marginal log-linear parameterizations. Algorithms for estimating marginal models are discussed, focussing on the maximum likelihood and the generalized estimating equations approaches.
It is shown how marginal models can help to understand directed graphical and path models, and a description is given of marginal models with latent variables.
\end{abstract}

\section{Introduction}

We start with motivating examples in Section 2, including repeated measurements, missing data, and graphical models. All these involve the application of models which apply restrictions only on subsets of the variables, that is, on marginals of the contingency table containing their joint distribution.

The restrictions imposed by marginal models apply to the association structures within subsets of variables. The association is captured by log-linear parameters calculated in marginals of the table and Section 3 deals with general aspects of parameters and parameterizations, including variation independence.

Marginal log-linear parameterizations are developed in Section 4 and some of their fundamental properties, like variation independence, smoothness, and collapsibility are also discussed. Depending on the choice of the marginals in which the log-linear parameters are determined, marginal log-linear parameterizations are appropriate to capture several characteristics of the marginal and conditional association structure.

Marginal log-linear models are defined by restricting some marginal log-linear parameters to zero, as described in Section~5.

Marginal log-linear parameters are the standard log-linear parameters calculated from marginal distributions and measure the strength of conditional and/or marginal association. Marginal log-linear parameters are based on ordinary odds and odds ratios and their higher-dimensional generalizations. But, as described in Section 6, other types of odds ratios, which are particularly useful for ordinal data, may also be used to define marginal log-linear models.

Section 7 contains results concerning the general type of conditional independence models, including the case when some conditional independences apply to subsets of the variables, which may be formulated as marginal log-linear models.

Section 8 deals with estimation and testing. Lagrangian and Fisher scoring methods for maximum likelihood estimation are described and compared. 
The generalized estimating equations (GEE) approach for estimating marginal models is described as well. 

Section 9 discusses areas of applications where marginal log-linear models either provide a general way of implementing the standard analysis or a new approach to answer the research question. These include directed graphical models, path models, and latent variable models, but many other applications are mentioned, too.

Very few proofs are included, as most of the results are quoted from research publications.

 \section{Motivation}
 
 There are several types of statistical problems where marginal distributions of higher dimensional joint distributions play a central role. In this section, we discuss three such broad types of problems.

\subsection{ 
Repeated measurements and panel studies
}\label{repmeas}
In many experimental and observational settings, subjects are measured or observed repeatedly. The reasons for measuring repeatedly include to study the within-subject variability of the measurements, or to reduce measurement error by taking the average measurement value. In such cases, the measurements are made close to each other in time. Another reason for repeated measurements is  to investigate the effect of a treatment applied to the subjects between the measurements, in which case one measurement is taken before, and another one after, the treatment. Sometimes the variability or stability of the measurement results over time is of interest, without any treatment being applied.

For example, variables $A_1$ and $B_1$ are observed in a first measurement, a treatment is applied, and then the same variables are measured again, denoted as $A_2$ and $B_2$. There are a number of relevant hypotheses to test. The first one, say $H_1$, is that $A$ and $B$ are independent both before and after treatment. One may argue that $H_1$ is true if and only if both $H_{11}:$ ``$A_1$ is independent of $B_1$'' and $H_{12}:$ ``$A_2$ is independent of $B_2$'' are true. This is correct, but a test of $H_1$ with a given level cannot be  constructed, in general, from separate tests of the hypotheses $H_{11}$ and $H_{12}$. This would be possible if the samples for the pairs of variables $A_1, \, B_1$ and $A_2, \, B_2$ were independent, which is not the case in the current repeated measurements setup. Instead, one has observations for each unit in the sample for the variables $A_1, \, B_1, \, A_2, \, B_2$, and $H_1$ states that in this $4-$dimensional distribution there are two marginal independences, one for $A_1$ and $B_1$, and one for $A_2$ and $B_2$. This is a marginal model.

Another relevant model, say $H_2$, in this setup is that the distributions of the two measurements of $A$ are identical, that is, the treatment does not change the distribution on the population level, and similarly for $B$, but the results of the second measurements are independent. Thus $H_2$ contains restrictions on the $A_1 \times A_2$ (marginal homogeneity), $B_1 \times B_2$ (marginal homogeneity), and $A_2 \times B_2$ (independence) marginals.

The hypotheses $H_1$ and $H_2$ assume marginal models about the joint distribution.

A closely related longitudinal design is called a panel study, see, e.g., Frees  and  Kim (2008), where the individuals in a sample  are interviewed repeatedly at regular intervals. The advantage\footnote{The design also has disadvantages, of course. These include panel attrition and the fact that, even if originally selected appropriately, with passing time the sample will become different in composition from the current population.} of such a design is that changes in opinions, preferences, or attitudes may be studied in a more valid way than by simultaneously asking about current and also previous positions in a cross-sectional study. The main limitation of such an approach is that earlier opinions or attitudes are often not remembered and sometimes are not reported truthfully.

In the analysis of panel data, the transition probabilities from one position into another one are of central interest. In particular, the dependence of the transition probabilities on earlier positions is an important question because this determines the  fragmentation of the data. More precisely, if the panel has, say, $5$ waves and $A_1, \,A_2,\, A_3,\, A_4,\, A_5$ denotes the positions of a respondent regarding a particular question during the waves, then one is interested in deciding whether, for instance,
$$
P(A_5| A_4, A_3) = P(A_5| A_4) 
$$
holds. If it does, then  the position at wave $5$ cannot be better predicted if, in addition to the position at wave $4$, the position at wave $3$ is also taken into account. For example, in this case the chance of supporting a particular political party at the time of wave $5$ may depend on the preferred party at the time of wave $4$, but if the latter is known, the party preference at the time of wave $3$ provides no additional information.

Slightly more generally, if
$$
P(A_t| A_{t-1}, A_{t-2}, \ldots , A_1) = P(A_t| A_{t-1}) 
$$
holds for all waves (time points) $t$, then the joint distribution is called a one-step Markov chain. It is easy to see that this property is equivalent to the following conditional independence
$$
A_t \ind  A_{t-2}, \ldots , A_1 | A_{t-1}.
$$
In detail, for the $5$ waves this means that 
$$
A_3 \ind   A_1 | A_{2},
$$
$$
A_4 \ind   A_2, A_1 | A_3,
$$
$$
A_5 \ind   A_3, A_2, A_1 | A_4.
$$
For the joint distribution of the variables $A_1, \, A_2, \, A_3, \, A_4, \, A_5$ the model prescribes  conditional independences on the $A_1 \times A_2$, $A_1 \times A_2 \times A_3$,  $A_1 \times A_2 \times A_3 \times A_4$ and $A_1 \times A_2, \times A_3 \times A_4 \times A_5$ marginals.

\subsection{ 
Missing data and data fusion
}

The statistical problems discussed next lead to the task of generating a joint distribution with given marginal distributions. Thus, the restrictions implied by design or the type of data collected in these cases fully determine some marginal distributions (or make it possible to estimate them) and do not only specify a model for them as in the previous examples.  

One group of such problems is related to incomplete observations or missing data (Little  and  Rubin, 2019). When the data are collected through a survey of a human population, usually not all individuals selected by the sampling procedure answer the questions. Some are not found, some are found but are not willing to participate in the survey, and some do participate but choose not to answer some questions. While dealing with those who do not provide any information is a serious issue, the best utilization of the often only partial answers collected is an important statistical problem (Little  and  Rubin, 2019). A similar situation occurs when data are collected in an experimental setting, because of the dropout of the participants. One approach is to consider the responses collected for a particular subset of the questions and use them to estimate the joint distribution of the answers. These are estimates of some marginal distributions of the joint distribution of all answers. This procedure is justified, because the smaller a subset of questions is, the more individuals gave responses to all of them, and their joint distribution may be better estimated than the joint distribution of all variables.

For example, let the questionnaire contain $4$ yes-no questions and let the variables $A_1, \, A_2, \, A_3, \, A_4$ contain the  answers. Then, the distribution of $A_1$ may be estimated based on all the answers provided to the first question, and similarly for all other variables. Thus, the one-way marginal distributions are estimated based on different subsets of the sample. Next, the $A_1 \times A_2$ marginal distribution is estimated based on the one-way marginal estimates already obtained and on the observations which contained responses  to both $A_1$ and $A_2$. From the latter, one may estimate the odds ratio (see, e.g., Rudas, 2018) between $A_1$ and $A_2$ and combine this with the one-way marginals to estimate the $A_1 \times A_2$ distribution.  In theory, the procedure can be continued until the $A_1 \times A_2 \times A_3 \times A_4$ distribution is estimated, although it raises many compatibility and optimality issues, and as will be seen later, the feasibility of such a procedure depends heavily on the patterns of missing data.

Sometimes, the missing data pattern is not observed but is implied by design. When the questionnaire is too long, or answering all questions could be seen as a  breach of the respondents' privacy, some of them may be asked $A_1, \,A_2, \,A_3$, others $A_1, \, A_4, \, A_5$, where now these may not be individual questions rather blocks of questions, and, of course, other patterns are also possible. The $A_1 \times A_2 \times A_3$ and the $A_1 \times A_4 \times A_5$ marginal distributions may be estimated, and from these the joint distribution. The design is called a split questionnaire (Rhemtulla  and  Little, 2012)  but similar problems arise in so-called register-based censuses, see, e.g., Eppmann et al. (2006).

In a register-based census, now applied by several countries, instead of collecting information from all inhabitants of the country, data from existing registers (driving licences, health care access, etc.) are combined to find out the relevant information. The individual registers provide certain conditional and/or marginal distributions, and the task is to estimate the joint distributions. This problem is called data fusion (see, e.g., D'Orazio et al., 2006), and it also occurs in other areas, see e.g., Cocchi (2019).

\subsection{ 
Graphical modelling
}\label{grm}

Graphical Markov models associated with directed acyclic graphs (also called Bayesian nets) are widely used in expert systems, artificial intelligence, and also in some approaches to modelling causal effects. 

A simple example of a directed acyclic graph (DAG) is shown in Figure \ref{fig1}. It has four nodes, $A, B, C, D$, which are identified with variables and the intuitive interpretation of the arrows is that they represent direct effects. The graph is acyclic, because there is no sequence of nodes in the order of arrows with the same starting and ending node.

\begin{figure}
\caption{A directed acyclic graph}
\vspace{-30mm}
\includegraphics[trim=0 600 0 0, clip]{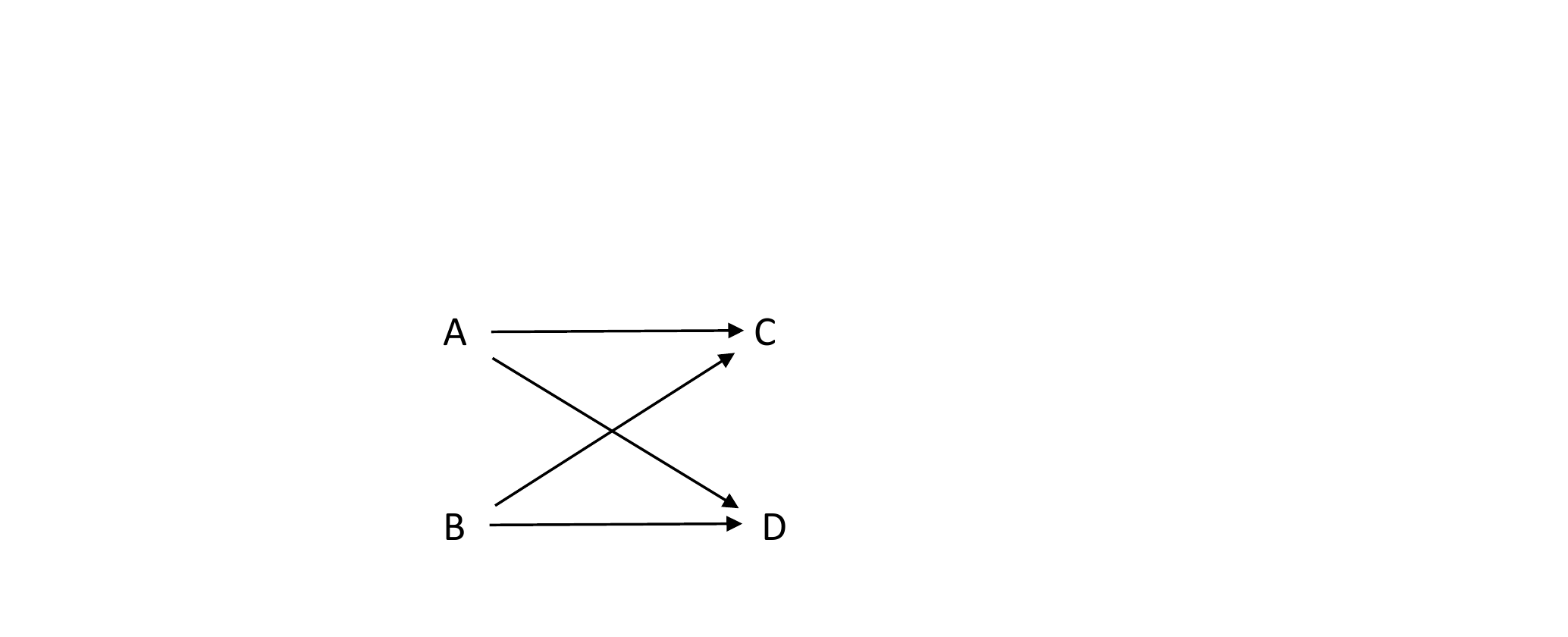}
\label{fig1}
\end{figure}

A precise interpretation of the Markov model is that it assumes (conditional) independences among the variables. It is the missing arrows
which imply the conditional independences defined by the directed Markov property, see, e.g., Lauritzen (1996).\footnote{See Bergsma et al.\ (2009), Chapter~5, for applications in causal analysis and the relationship with structural equation models.} 

Graphical Markov models associated with DAGs  generalize the conditional independence property of Markov chains. In a Markov chain, the conditional independence  is implied by the lack of temporal adjacency, and in the more general structures considered here, the temporal adjacency is replaced by the adjacency read off from a graph. When this graph is undirected, one Markov property, called the local Markov property, says that a variable is conditionally independent of its non-neighbours, given its neighbours.\footnote{The neighbours of a node $A$ are those nodes with which $A$ is connected by an edge.} Notice that such a conditional independence involves all variables.

The assumption that the joint distribution of the variables obeys the conditional independences implied by the local Markov property applied to a particular graph is a graphical log-linear model, see Lauritzen (1996). 

To define the Markov property associated with a DAG, call those nodes from which an arrow goes to $A$ the parents of node $A$ and denote them as $\pa(A)$. Further, call those nodes into which no directed path leads which starts in $A$ the non-descendants of $A$ and denote these nodes as $\nd(A)$. Note that $\pa(A) \subseteq \nd(A)$, because otherwise the graph would contain a directed cycle. Then the local Markov property is that 
$$
A \ind \nd(A) \setminus \pa(A) | \pa(A).
$$
In the case of the DAG in Figure \ref{fig1}, the directed Markov property implies that 
\begin{equation}
\label{ind}
A \ind B
\end{equation}
\begin{equation}
\label{cind}
 C \ind D\, | \, A, B.    
\end{equation}

Out of the two (conditional) independences, one is on the  $A \times B$ marginal, and the other one is on the full table.

In general, the conditional independences defining a Markov model associated with a DAG are on marginals containing a variable and its non-descendants. Therefore, these models are marginal models.

An important property of the distributions which are Markov according to a DAG is that they factorize into the product of conditional distributions of the variables given their parents, see, e.g., Lauritzen (1996). For the DAG in Figure \ref{fig1}, the factorization is $$
P(ABCD) = P(A)P(B)P(C|A,B)P(D|A,B).
$$

\section{Parameterizations of discrete probability distributions}
\label{sec-par}

A marginal model is defined most conveniently using a particular parameterization of the joint distribution of the variables of interest, the so-called marginal log-linear parameterization, to be discussed in Section~\ref{MLLP}. This section deals with various general characteristics of parameters and parameterizations. 

\subsection{ 
Parameters and parameterizations
}

A parameter is an arbitrary function of the distribution, and is often multidimensional, i.e.,  vector-valued. For example, in the case of a $2 \times 2$ distribution, with the usual notation, $(p_{11}, p_{12})$ is a parameter, so is $(p_{11}, p_{12}, p_{21}, p_{22})$ or $( p_{11}, p_{1+}, p_{+1}) $. The parameters which are of interest in statistical analysis usually express some relevant property of the distribution. For example, a widely used measure of association between the two variables forming the table, see e.g., Rudas (2018), is the odds ratio
$$
\frac{p_{11}p_{22}}{p_{12}p_{21}},
$$
which is also a parameter of the distribution. It measures a characteristic (strength and direction of association) which is not directly seen from the probabilities. 

Therefore, a parameter represents information from the distribution. In many cases, one is interested in looking at parameters that carry all information in the distribution. A more formal way of imposing this is to consider, instead of the value of a certain parameter, the function which yields that parameter and to require that this function is invertible. If this holds, the parameter is called a parameterization. 

In the case of a $2 \times 2$ distribution, the $(p_{11}, p_{12})$  parameter is not a parameterization, because if its value is known, the distribution cannot be reconstructed. But the $(p_{11}, p_{12}, p_{21}, p_{22})$ and $( p_{11}, p_{1+}, p_{+1}) $ parameters are parameterizations. The cell probabilities are given in the first case, and are easily determined in the second. Also, the odds ratio and the marginal probabilities $p_{1+}$ and $p_{+1}$ form a parameterization, see Rudas (2018). In this case, inverting the parameterization, i.e. calculating the cell probabilities, needs to be done using numerical algorithms such as  the Iterative Proportional Fitting or Scaling algorithm, see, e.g., Rudas (2018).

While a parameter vector may have arbitrary dimension, a parameterization has a minimal dimension, which is the dimension of the distribution. In the case of a $2 \times 2$ distribution, although one has $4$ probabilities, in the $4$-dimensional space the distributions are in a $3$-dimensional subspace, as their sum is $1$. The same fact may also be formulated by saying that out of the $4$ probabilities, only $3$ are linearly independent. Therefore, the minimal dimension of a parameterization is $3$.

\subsection{ 
Variation independence
}\label{par2}

One of the most desirable properties of parameters and parameterizations is the variation independence of their components. Before giving a general definition, a simple example adopted from Rudas  and  Bergsma (2004) is given to illustrate the concept.

Suppose in an experiment a $2 \times 2$ treatment by outcome table was observed, and
as a measure of  effect of the treatment,  the difference in
proportion of positive outcomes among those treated and among the control is used. Assume the data given in Tables \ref{tab1} and  \ref{tab2} 
 were observed for male and female participants, respectively.

\begin{table}
\[
\begin{array}{|r|r|r|r|}
\hline
  Outcome & Positive & Negative & Total \\
\hline
 Treatment & 20   &  80  & 100  \\
\hline
 Control & 10 & 90   & 100 \\
\hline
\end{array}
\]
\caption{Hypothetical experimental results for men }
\label{tab1}
\end{table}

\begin{table}
\[
\begin{array}{|r|r|r|r|}
\hline
  Outcome & Positive & Negative & Total \\
\hline
 Treatment & 60   &  40  & 100  \\
\hline
 Control & 40 & 60   & 100 \\
\hline
\end{array}
\]
\caption{Hypothetical experimental results for women }
\label{tab2}
\end{table}

The selected measure of the size of the effect takes on the value of 
$$
\frac{20}{100} - \frac{10}{100} = 0.1
$$
for men and is
$$
\frac{60}{100} - \frac{40}{100} = 0.2
$$ 
for women. Is, then, the treatment twice as effective for women than for men?

The answer to this question is not clear. While the difference in the probabilities  of positive outcomes under treatment and control seems like a meaningful measure of effect and its value is twice as big for women as for men, the following argument is also possible: for both men and women, $100$ individuals received the treatment and $100$ the control, but for men, there were $30$ and for women $100$ positive outcomes. Given this, the maximum possible value of the measure of effect is    
$$
\frac{30}{100} - \frac{0}{100} = 0.3
$$
for men and is
$$
\frac{100}{100} - \frac{0}{100} = 1
$$ 
for women. Thus, the actual value of the measure of treatment efficacy is $1/3$ of its maximum  possible value for men, while for women the actual value is only  $1/5$ of its theoretical maximum.

This example illustrates that the possible range of the measure is affected by the values of the other parameters and, in such a case, the assessment of the actual value may be very different when the other parameters are or are not taken into  account. Put differently, a parameter which is not variation independent of the other parameters lacks calibration.

Variation independence means that the above dependence does not occur. Two parameters are variation independent if their joint range is the Cartesian product of their individual ranges, i.e., any otherwise possible value of one can be combined with any otherwise possible value of the other.

In the case of the example discussed above, the measure of treatment efficacy was
\begin{equation}\label{effi}
\frac{p_{11}}{p_{+1}} - \frac{p_{12}}{p_{+2}}
\end{equation}
and this was not variation independent of the other parameters $p_{1+}$ and $p_{+1}$.
Indeed, its minimum value is zero and its maximum value is 

$$
\text{min} \left(  1, \frac{p_{1+}}{p_{+1}} \right) ,
$$
so its range is $[0,\text{min}(1, p_{1+}/p_{+1})]$. To put it differently, the range of the measure (\ref{effi}) depends on the marginal distributions, and it is often not clear whether the inference should or should not condition on the marginals.   

The odds ratio is variation independent of $(p_{1+}, p_{+1})$ and is, therefore, a parameter of the association with a calibration which does not depend on the marginals. Consequently, its values may be compared, even if calculated for tables with different marginal distributions. The value of the odds ratio is $1,800/800=2.25$ for men and $3,600/1,600=2.25$ for women, suggesting that the strength of association between treatment and positive outcome is the same for men as for women, in contrast with the naive comparison of the measures calculated originally. This is, however, not to say that the odds ratio would be without problematic characteristics when used as a measure of treatment efficacy, see Rudas (2010).

For higher dimensional  tables, there are parameterizations which rely on the odds ratio and its generalizations with variation independence properties, leading to a natural definition of log-linear models, see Rudas (2018), but for the definition of marginal models another type of parameterization, based on marginal and conditional distributions, is of more immediate use.

In the case of a $2 \times 2$ distribution, a parameterization with the $p_{+1}$ marginal probability and the $p_{1|1}$ and $p_{1|2}$ conditional probabilities is also possible. Indeed,
$$
p_{11}= p_{1|1} p_{+1}
$$
$$
p_{21}= (1-p_{1|1}) p_{+1}
$$
$$
p_{12}= p_{1|2} (1-p_{+1})
$$
$$
p_{21}= (1-p_{1|1}) (1-p_{+1}).
$$
An important feature of this parameterization is that all three parameters in it are variation independent.

A similar parameterization of a distribution on a $4$-way $A \times B \times C \times D$ table may parameterize the distribution on the $A \times B$ marginal, and then parameterize the conditional distribution on $C \times D$, given the marginal distribution on $A \times B$. Here, the two groups of parameters are variation independent. Further, within this parameterization, one may impose the marginal independence of $A$ and $B$, and then the conditional independence of $C$ and $D$, given $A$ and $B$. Note that this is exactly the marginal model defined in (\ref{ind}) and (\ref{cind}), as implied by the Markov property applied to the graph in Figure \ref{fig1}.

Alternatively, the following parameters constitute a parameterization of the $4$-way table in the binary case, with $OR$ denoting  the odds ratio and $COR$ the conditional odds ratio:
$$
\theta_1=(P(A=1), P(B=1)
$$
$$
\theta_2= \OR(A,B)
$$
$$
\theta_3 =(P(C=1|A=1, B=1),P(C=1|A=1, B=2), $$ $$ P(C=1|A=2, B=1),  P(C=1|A=2, B=2), P(D=1|A=1, B=1),  $$ $$ P(D=1|A=1, B=2), P(D=1|A=2, B=1), P(D=1|A=2, B=2 )) 
$$
$$
\theta_4 = (\COR(C, D |A=1, B=1), \COR(C, D |A=1, B=2),$$ $$ \COR(C, D |A=2, B=1), \COR(C, D |A=2, B=2)).
$$

In this example, $\theta_1$ is equivalent to the marginal distributions of variables $A$ and of $B$ and $\theta_3$ gives the conditional distributions of $C$ and of $D$, given any possible category combinations of $A$ and $B$. The parameter $\theta_2$ is the odds ratio in the marginal distribution  $A \times B$, and $\theta_4$ is the collection of the conditional odds ratios of $C$ and $D$, given all possible category combinations of $A$ and $B$. Thus, $\theta_1$ and $\theta_2$ determine the $A \times B$ marginal distribution, and  $\theta_3$ and $\theta_4$ determine the $C \times D$ conditional distribution, given $A \times B$.

Here, $\theta_1$ and $\theta_2$ are variation independent. Further, $\theta_3$ is variation independent of $\theta_1$ and $\theta_2$, and $\theta_4$ is variation independent of all the other three parameters. But also, $\theta_2$ and $\theta_4$ are variation independent of $\theta_1$ and $\theta_3$.

A statistical model may be obtained by fixing the values of some parameters of a parameterization. If variation independence between the fixed and the other parameters holds, this has no implication for the possible values of the other parameters, that is, there will be exactly one distribution in the model for every possible value of the other parameters. In other words, the unrestricted (components of the) parameter parameterize the model obtained by restricting the other (components of the) parameter.

This is illustrated most easily with the parameterization of  a $2 \times 2$, $A \times B$ table with  $\eta_1= (p_{1+}, p_{+1})$ and $\eta_2$ the odds ratio $\OR(A,B)$. As $\eta_1$ and $\eta_2$  are variation independent, if one defines a model by imposing $\eta_2=1$, i.e., the independence of $A$ and $B$, then there is exactly one independent distribution for every choice of the marginal probabilities in $\eta_2$.  Related models are obtained by fixing the odds ratio at a different value, see Rudas (1991) and Rudas  and  Leimer (1992).


In the case of the $2 \times 2$  example above, when $\eta_2$ is fixed at $1$ and one obtains the model of independence for the $2 \times 2$ table, the number of degrees of freedom is $1$. It may sound counter-intuitive that when more parameters are fixed by the model, then the number of degrees of freedom is higher. To accept this, one has to remember that the number of degrees of freedom is related to the amount of deviation between observed and expected frequencies tolerated before one would decide the data provide evidence against the model. The more restrictive is the model, the larger is the deviation between observed and expected frequencies that one is ready to tolerate without rejecting the model. Thus, to have  a testing procedure with fixed type I error probability, one wishes to use  critical values which increase monotonically with the number of parameters fixed by the model, and chi-squared distributions with larger  degrees of freedom have larger  critical values. Therefore, the number of degrees of freedom associated with a model is not related to the parameters left free, rather it is related to (more precisely, is equal to) the number of parameters fixed by the model. 

In the case of the $4$-dimensional example, setting $\theta_2=1$ implies (\ref{ind}) and setting $\theta_4=(1,1,1,1)$ implies (\ref{cind}), yielding the graphical model associated with the graph in Figure \ref{fig1}. These two parameters are variation independent of the other two parameters, thus, $\theta_1$ and $\theta_3$ parameterize all distributions which are Markov according to the graph in Figure \ref{fig1}. Further, as there are $5$ parameter values fixed by the model ($\theta_2$ and $\theta_4$), the standard Pearson and likelihood ratio statistics have an asymptotic chi-squared distribution on $5$ degrees of freedom, when the model holds true, and data and maximum likelihood estimates are compared to test model fit. 

Following Bergsma  and  Rudas (2002a), marginal models will be defined in this chapter as a generalization of the procedure above.

 \section{Marginal log-linear parameterizations}
 \label{MLLP}
 
 Marginal log-linear parameters and parameterizations generalize the example discussed for the $4$-way table in the last section. Marginal  models will be defined by setting some marginal log-linear parameters to zero. The definition of marginal log-linear parameters, and of marginal log-linear models, provides flexible applications which can capture various useful properties of the joint distribution of several categorical variables.

 \subsection{ 
 Definition
 }\label{def}
 
 Let $\mathcal{V}$ be a set of categorical variables, and let $\mathcal{M} \subseteq \mathcal{V}$ denote a marginal. In the sequel, the word marginal will be used, depending on the context, as a subset of the variables, a marginal table of the associated contingency table, or the marginal distribution derived from a joint distribution.
 
 Following Bergsma  and  Rudas (2002a), marginal log-linear parameters are defined as log-linear parameters  (see e.g., Bishop et al. (1975), Agresti (2013), Rudas (2018)), calculated in  marginals of the table. For simplicity, only distributions with positive cell probabilities are considered in this chapter.
 
 Every subset $\mathcal{E} \subseteq \mathcal{V}$ of the variables may have an effect associated with them, which affects the joint distribution. To emphasize this, subsets of the  variables will also be referred to as effects. The strength of the effect (associated with a subset of variables) may be quantified in different ways. A particular quantification is, of course, a parameter. In this section a particular choice of the parameters is used, and alternatives will be discussed later. Many of the properties of the models do not depend on the particular choice of the parameters; however, this becomes relevant when estimated parameter values are used to describe distributions in the model. 
 
 A classical log-linear parameter (see e.g., Bishop et al. (1975), Agresti (2013), Rudas (2018)) for an effect,  $\mathcal{E}$ associates a value with every  category combination $e$ of the variables $\mathcal{E}$, denoted as $\lambda^{\mathcal{E}}_e$. These parameters are defined via the following recursion:
 \begin{equation} \nonumber
\lambda^{\emptyset} =      \frac{1}{c_{\mathcal{V}}} \sum_v \log P(v),
 \end{equation}
 \begin{equation}
\lambda^{\mathcal{E}}_{e} = \frac{1}{c_{\mathcal{V}\setminus \mathcal{E}}} \sum_{v: (v)_{\mathcal{E}}=e} \log P(v) - \sum_{\mathcal{F}\subsetneq \mathcal{E}} \lambda^{\mathcal{F}}_{(e)_{\mathcal{F}}}
\end{equation}
where $e$ is a joint category of  the variables $\mathcal{E}$,  $c_{\mathcal{V}\setminus \mathcal{E}}$ denotes the number of joint categories of the variables in $\mathcal{V}  \setminus \mathcal{E}$, and $(v)_{\mathcal{E}}$ denotes the categories out of $v$ which belong to the variables in~$\mathcal{E}$.

 When all the variables are binary, the log-linear parameters can be shown to be equal to various  averages of the  logarithms of the roots of $(l-1)$th order conditional odds ratios of the $l$ variables  in $\mathcal{E}$, given all possible category combinations of the variables $\mathcal{V}\setminus \mathcal{E}$, (see, e.g., Rudas (2018) and Section 6 of this chapter). For example, in a binary $A \times B \times C$ table
 $$
 \lambda^{AB}_{21} = \frac{1}{2}  \left( \log P(2,1,1)+\log P(2,1,2)  \right)
 $$
 $$
 - \frac{1}{4} \left( \log P(2,1,1)+ \log P(2,1,2)+ \log P(2,2,1)+ \log P(2,2,2)    \right)
 $$
 $$
 - \frac{1}{4} \left( \log P(1,1,1)+ \log P(1,1,2)+ \log P(2,1,1)+ \log P(2,1,2)    \right)
 $$
 $$
 + \frac{1}{8} \left(\log P(1,1,1)+ \log P(1,1,2)+ \log P(1,2,1)+\log P(1,2,2) \right)
 $$
 $$
 + \log  \left( P(2,1,1)+\log P(2,1,2)+\log P(2,2,1)+\log P(2,2,2)    \right)
 $$
 $$
 =  \log \sqrt[8]{\frac{P(1,2,1)P(1,2,2)P(2,1,1)P(2,1,2)}{P(1,1,1)P(1,1,2)P(2,2,1)P(2,2,2)}}  
 $$
 $$
 =\frac{1}{2}  \left( \log \sqrt[4]{\frac{P(1,2,1)P(2,1,1)}{P(1,1,1)P(2,2,1)}} + \log \sqrt[4]{\frac{P(1,2,2)P(2,1,2)}{P(1,1,2)P(2,2,2)}}      \right)
 $$
 \begin{equation} \label{LLpar}
 = \frac{1}{2} \left( \log \sqrt[-4]{COR(A,B|C=1)}+ \sqrt[-4]{COR(A,B|C=2)}  \right).
 \end{equation}
 
 In general, the log-linear parameter for every effect will be considered as vector-valued, with one component for every category combination of the variables in $\mathcal{E}$, except for the combinations where any of the variables is in, say, its first category, in order to avoid linear dependence of the components of the parameter. So if $\mathcal{E} = \{ V_1, V_2, \ldots, V_l \}$, and these variables have $c_1, c_2, \ldots, c_l$ categories, then the log-linear parameter has
 \begin{equation}\label{numb}
 (c_1-1)(c_2-1) \dotsi (c_l-1)
 \end{equation}
 components and these components are, in general, linearly independent.

The $(l-1)$th order conditional odds ratio of the variables in
 $\mathcal{E}$ (when conditioned on any category combination of the variables $\mathcal{V} \setminus \mathcal{E}$) is variation independent of the marginal distributions of the variables in any proper subset of $\mathcal{E}$, see, e.g., Rudas (2018). This was illustrated above for the simple cases of $2$- and $4$-way tables. Therefore, the log-linear parameters, which are functions of the conditional odds ratios, are  widely used as measures of the amount of association within an effect, that cannot be attributed to a proper subset of the variables in the effect.

 Calculating the log-linear parameter for $\mathcal{E}$ in a marginal $\mathcal{M}$, with  $\mathcal{E} \subseteq \mathcal{M}$, means that the marginal probabilities of $\mathcal{M}$ are used, instead of the joint probabilities of $\mathcal{V}$. For example, in a $4$-way binary $A \times B \times C \times D$ table, in the $A \times B \times C$ marginal, the value of the marginal log-linear parameter for the $AB$ effect is
 \begin{equation} \label{MLL}
  \frac{1}{2} \sum_{k=1}^2 \log \left( \frac{P(1,1,k,+)P(2,2,k,+)}{P(1,2,k,+)P(2,1,k,+)}\right)^{1/4}.
 \end{equation}
 The value in (\ref{MLL}) is denoted as $\lambda^{ABC}_{AB}$ and in general as $\lambda^{\mathcal{M}}_{\mathcal{E}}$. 
 The parameter $\lambda^{ABC}_{AB}$ 
 is a measure of average (over categories of $C$) conditional association between variables $A$ and $B$, (with averaging over the categories of $C$), calculated in the $A\times B \times C$ marginal of the four-way distribution.

 The parameter $\lambda^{ABC}_{AB}$ 
 has a single value, as both $A$ and $B$ are binary. If, for instance, $B$ has three categories, so the table is of the size $2 \times 3 \times 2 \times 2$, then $ \lambda^{ABC}_{AB}$ has $2$ components, one for the $(2,2)$ and one for the $(2,3)$ indices of $A$ and $B$. Out of these, the one associated with $(2,2)$ is as given in (\ref{MLL}), and the one associated with $(2,3)$ depends on the type of odds ratio selected, see Section
\ref{altpar}. Such choices are governed by the characteristics of the research question.

 Let 
 $$
 \mathcal{M}_1, \mathcal{M}_2, \ldots, \mathcal{M}_k
 $$
 be a sequence of marginals, such that
 $$
 \mathcal{M}_j \nsubseteq \mathcal{M}_i, \makebox{ if } i < j
 $$
 and
 $$
 \mathcal{M}_k = \mathcal{V}.
 $$
 Such a sequence will be called {\it non-decreasing}. Marginal log-linear parameters calculated in such sequences of marginals play a central role in this chapter.

To define the marginal log-linear parameters, for every effect $\mathcal{E}$, let $\mathcal{M}(\mathcal{E})$ be the first marginal in the non-decreasing order, that contains it:
 \begin{equation}\label{mdef}
   \mathcal{M}(\mathcal{E}) = \mathcal{M}_i  \makebox{ if }\mathcal{E} \subseteq \mathcal{M}_i \makebox{ and } \mathcal{E} \nsubseteq \mathcal{M}_j \makebox{ if } j  < i.
\end{equation}
 
 Let now $\lambda^{\mathcal{M}(\mathcal{E})}_{\mathcal{E}}$ denote the log-linear parameter of the effect $\mathcal{E}$ calculated within the 
 $\mathcal{M(\mathcal{E})}$ marginal. This is a log-linear parameter calculated not in the joint distribution of all variables, but rather in a marginal distribution. As illustrated above, the parameter is usually vector-valued, but this fact will be suppressed in the sequel. Marginal log-linear parameters measure the strength of marginal and conditional associations at the same time. By the choice of the marginal, in which the parameter for an effect is defined, some variables are disregarded, and then one conditions upon the variables which are in the marginal but do not belong to the effect. For further discussion of conditional and marginal association, see Bergsma  and  Rudas (2013). The marginal log-linear parameters $\{\lambda^{\mathcal{M}(\mathcal{E})}_{\mathcal{E}}, \, \mathcal{E} \subseteq \mathcal{V}\}$ are called {\it hierarchical and complete}.
 
The marginal log-linear parameters, as defined here, obviously contain as special cases the ordinary log-linear parameters, but also the multivariate logistic parameters (McCullagh  and  Nelder, 1989, Glonek  and  McCullagh, 1995) as well as  the mixed parameters considered in Glonek (1995).

 Note that the parameters defined at the end of the previous section are one-to-one functions of marginal log-linear parameters. In this example, $\mathcal{V}= \{A, B, C, D  \}$, $\mathcal{M}_1= \{A, B  \}$, and   $\mathcal{M}_2= \{A, B, C, D  \}$. Thus $\mathcal{M}(\emptyset)=\mathcal{M}(A)=\mathcal{M}(B)=\mathcal{M}(A,B)=\mathcal{M}_1$ and $\mathcal{M}(C)=\mathcal{M}(A, C)=\mathcal{M}(B,C)=\mathcal{M}(A,B,C)=\mathcal{M}(D)=\mathcal{M}(A, D)=\mathcal{M}(B,D)=\mathcal{M}(A,B,D)=\mathcal{M}(C,D)=\mathcal{M}(A,C,D)=\mathcal{M}(B,C,D)=\mathcal{M}(A, B,C,D)=\mathcal{M}_2$. The parameters specified are  one-to-one functions of the marginal log-linear parameters of the effects. In particular, setting $\theta_2$ and $\theta_4$ equal to $1$ is the same as setting
 $$
 \lambda^{AB}_{AB} = 0 \makebox{ and } \lambda^{ABCD}_{CD} = 0. 
 $$

 In order to obtain analytical properties of marginal log-linear parameters, it is important to remove from among them those which are redundant in the sense that they can be calculated from the others; this is assumed to be the case throughout the whole chapter. The formula in (\ref{numb}) only took the non-redundant values of the parameter into account.

 \subsection{ 
 Basic properties
 }

 Marginal log-linear parameters have a number of desirable properties. With $\cM$ defined by~(\ref{mdef}):
 
 \begin{theorem}\label{par}
 The parameters $\{\lambda^{\mathcal{M}(\mathcal{E})}_{\mathcal{E}}: \mathcal{E} \subseteq \mathcal{V}  \}$ 
 constitute a parameterization of the joint distribution of the variables $\mathcal{V}$.
 \end{theorem}
 \begin{proof}
 This is part of Theorem 2 in Bergsma  and  Rudas (2002a). Technically, the proof is based on repeated applications of the Iterative Proportional Scaling procedure to determine joint distributions based on mixed parameterizations of exponential families, see, e.g., Rudas (2018). The general relevant result on mixed parameterizations is given by Barndorff-Nielsen (1978).
 \end{proof}

 The argument in the proof above also implies that for any $1 \le i \le k$, the marginal log-linear parameters calculated in $\mathcal{M}_1, \ldots 
 \mathcal{M}_i $ can be used to determine the joint distribution of the variables in $\mathcal{M}_i$. This implies the following result.

  \begin{theorem}\label{conddistr}
 If $\mathcal{M}_i \setminus \cup_{j < i} \mathcal{M}_j \neq \emptyset $, then  the marginal log-linear parameters
$\{\lambda^{\mathcal{M}(\mathcal{E})}_{\mathcal{E}}:\mathcal{M}(\mathcal{E}) = \mathcal{M}_i\}$
 determine the conditional joint distribution of the variables in $\mathcal{M}_i  \setminus \cup_{j < i} \mathcal{M}_j $, given the joint distributions of variables in $\mathcal{M}_i \cap( \cup_{j < i} \mathcal{M}_j)$.
 \end{theorem}

To illustrate Theorem \ref{conddistr}, for the variables $A$, $B$, $C$ let
$\mathcal{M}_1= \{AB\}$ and $\mathcal{M}_2= \{ABC\}$. Then $\mathcal{M}_2 \setminus \mathcal{M}_1 = \{ C \}$, and the effects which have their marginal log-linear parameters calculated in $\mathcal{M}_2$ are $C$, $AC$, $BC$, $ABC$ and they parameterize the conditional distribution of $C$, given the joint distribution  $AB$. As the joint distribution of $AB$ is parameterized in the marginal $\mathcal{M}_1$, the marginal log-linear parameters in the two marginals parameterize the $ABC$ joint distribution. The marginal log-linear parameters determined in the two marginals are variation independent.

For a less straightforward example, let $\mathcal{M}_1=\{A\}$, $\mathcal{M}_2=\{B\}$, and $\mathcal{M}_3=\{ABC\}$. In this case, the theorem is about the conditional distribution of $C$, given $AB$, but now the marginal log-linear parameters determined in  $\mathcal{M}_1=\{A\}$ and $\mathcal{M}_2=\{B\}$ do not determine the $AB$ joint distribution, only its $1$-way marginal distributions. In this case, out of the marginal log-linear parameters determined in $\mathcal{M}_3$, those belonging to the effects  $C$, $AC$, $BC$, $ABC$ determine the conditional distribution. This is most easily seen by including the $AB$ marginal as the third one, which would not change these parameters but would make the setup essentially the same as in the previous example, because in this case the $AB$ joint distribution would be parameterized before the parameters in the $\{ABC\}$ marginal are calculated.

If, however,
$\mathcal{M}_1= \{AB\}$, $\mathcal{M}_2= \{AC\}$, and $\mathcal{M}_3= \{ABC\}$, then {\color{black}$\mathcal{M}_3 \setminus \big(\mathcal{M}_1 \cup \mathcal{M}_2\big) = \emptyset $}, and Theorem \ref{conddistr} does not apply. Indeed, if conditioned on {\color{black}$AB\cup AC$}, no conditional distribution remains. What do the marginal log-linear parameters given in $\mathcal{M}_3$, which are for the effects $BC$ and $ABC$, determine? In this case, they determine the parameters which are needed in addition to the 
$AB$ and $AC$ marginal distributions to parameterize the $ABC$ distribution: the second-order odds ratio of $ABC$ and the conditional odds ratio of $B$ and $C$, given $A$, see, e.g., Rudas (2018).

 It is not true, in general, that all components of a marginal log-linear parameterization would be variation independent. The following result gives a necessary and sufficient condition for the components of a hierarchical and complete marginal log-linear parameterization to be variation independent.
 
 \begin{theorem}\label{orddec}
 The components of a hierarchical and complete marginal log-linear parameterization based on a non-decreasing sequence of marginals
 $$
 \mathcal{M}_1, \mathcal{M}_2, \ldots, \mathcal{M}_k = \mathcal{V}
 $$
 are variation independent, if and only if the following condition holds. Either   $k=2$ or for every $j=3, \ldots , k$, the maximal elements out of
 $$
 \mathcal{M}_1, \mathcal{M}_2, \ldots, \mathcal{M}_j,
 $$
 say
 $$
 \mathcal{H}_1, \mathcal{H}_2, \ldots, \mathcal{H}_l
 $$
 are such that either $l=2$ or for every  $3 \le h \le l$, there is $1 \le g=g(h) \le h-1$, such that 
 $$
     \left(\mathcal{H}_1 \cup \ldots \cup \mathcal{H}_{h-1}   \right) \cap \mathcal{H}_h = \mathcal{H}_g \cap \mathcal{H}_h.
 $$
 
 \end{theorem}
 \begin{proof}
 This is Theorem 4 in Bergsma  and  Rudas (2002a).
 \end{proof}
 
 The property formulated in the previous theorem is called {\it ordered decomposability}. If the marginals $
 \mathcal{M}_1, \mathcal{M}_2, \ldots, \mathcal{M}_k
 $ are all incomparable with respect to inclusion, thus all are maximal, then ordered decomposability means the standard decomposability concept, see, e.g., Rudas (2018).
 
 For example, in the case discussed last, with $\mathcal{M}_1= \{AB\}$, $\mathcal{M}_2= \{AC\}$, and $\mathcal{M}_3= \{ABC\}$, ordered decomposability holds. But if $\mathcal{M}_1= \{AB\}$, $\mathcal{M}_2= \{AC\}$, $\mathcal{M}_3=\{BC\}$, and $\mathcal{M}_4= \{ABC\}$, ordered decomposability does not hold, and it is easy to find values of the marginal log-linear parameters defined in $\mathcal{M}_1$ and $\mathcal{M}_2$, which restrict the range of the parameters in  $\mathcal{M}_3$; see Bergsma  and  Rudas (2002a). The  three $2$-way marginal (frequency) distributions presented in Table \ref{incomp}  are weakly compatible but not strongly compatible, that is, although the generated $1$-way marginals are all uniform, there is no $3$-way distribution with these marginals.

 \begin{table} 
\hspace{13mm}
                \begin{tabular}{|c|c|c|}
                    \hline
                      & B=1 & B=2 \\ \hline
                    A=1 & 3 & 1 \\ \hline
                    A=2 & 1 & 3  \\ \hline
                \end{tabular}
\hspace{10mm}
                \begin{tabular}{|c|c| c|}
                    \hline
                      & C=1 & C=2 \\ \hline
                    A=1 & 1 & 3 \\ \hline
                    A=2 & 3 & 1  \\ \hline
                \end{tabular}
\hspace{10mm}
                \begin{tabular}{|c |c |c|}
                    \hline
                      & C=1 & C=2 \\ \hline
                    B=1 & 3 & 1 \\ \hline
                    B=2 & 1 & 3  \\ \hline
                \end{tabular}
                \caption{Marginal distributions which are weakly compatible but not strongly compatible    }
            \label{incomp}
            \end{table}

 Indeed, if one had such a distribution, one would have for the frequencies that ${f(1,1,2) \le 1}$ (from the $BC$ marginal) and $f(1,2,2) \le 1$ (from the $AB$ marginal), but the sum of these two frequencies would have to be $3$ (from the $AC$ marginal). This means that the three $2$-way marginals, and consequently the corresponding marginal log-linear parameters, are not variation independent.
 
 This is an important difference between the standard and the marginal log-linear parameters. If the log-linear parameterization is calculated in the $\{ABC\}$ table, that is one has standard log-linear parameters, the parameter belonging to the $BC$ effect is essentially the conditional odds ratio $\COR(B,C|A=a)$ and this is variation independent of the $AB$ and $AC$ marginal distributions. But if a marginal log-linear parameterization is considered based on the marginals $\{AB\}$, $\{AC\}$, and $\{BC\}$, then the parameter belonging to the $BC$ effect is the marginal odds ratio $\OR(B,C)$ and this is not variation independent of the $AB$ and $AC$ marginal distributions.

  \begin{table} 
\hspace{33mm}
                \begin{tabular}{|c|c|c|}
                    \hline
                A=1      & C=1 & C=2 \\ \hline
                    B=1 & t & 3-t \\ \hline
                    B=2 & 1-t & t  \\ \hline
                \end{tabular}
\hspace{10mm}
                \begin{tabular}{|c|c| c|}
                    \hline
                    A=2  & C=1 & C=2 \\ \hline
                    B=1 & u & 1-u \\ \hline
                    B=2 & 3-u & u  \\ \hline
                \end{tabular}
 %
                \caption{Structure of a distribution with $AB$ and $AC$ marginals as given in Table \ref{incomp} }
            \label{condtable}
            \end{table} 
 
 To have the $AB$ and $AC$ marginal distributions as prescribed in Table \ref{incomp}, the $3$-way table has to have the  structure
 shown in Table \ref{condtable}, implying that $t \le 1$ and $u \le 1$. The conditional odds ratios are 
 $$
 \COR(B,C|A=1)= \frac{t^2}{(1-t)(3-t)}
 $$
 and
 $$
 \COR(B,C|A=2)= \frac{u^2}{(1-u)(3-u)}
 $$
 and their values are not restricted, i.e., depending on $t$ and $u$, may be anywhere on the interval $(0, \,\, \infty)$. But the marginal odds ratio is
 $$
 \OR(B,C)= \frac{(t+u)^2}{(4-t-u)^2}
 $$
 and this is restricted to be not more than $1$.

 However, even in this case, the marginal log-linear parameters calculated in the marginals  $\{AB\}$, $\{AC\}$, and $\{BC\}$  on the one hand, and the parameters calculated in $\{ABC\}$, on the other hand, are variation independent.

 \subsection{Smoothness of marginal log-linear parameters} \label{sec-smooth}

Marginal log-linear models will be defined by assuming that some marginal log-linear parameters are zero. Many of the statistical properties of these models, including the behaviour of maximum likelihood estimates and asymptotic distributions of test statistics depend on analytical properties of the parameterizations used.

A parameter is called smooth if, as a function of the (probability or frequency) distribution, it is continuous, invertible, twice continuously differentiable, and its Jacobian has full rank everywhere.

\begin{theorem}
\label{smthness}
The hierarchical and complete marginal log-linear parameters  are a smooth parameterization of the frequency distribution.
\end{theorem}
\begin{proof}
This is Theorem 2 in Bergsma  and  Rudas (2002a). Note that smoothness holds only if the redundant parameter values are omitted.
\end{proof}
To obtain a smooth parameterization of the probability distribution, the parameter referring to the empty set, $\lambda_\emptyset^{\cM(\emptyset)}$, must be omitted because its value is determined by the other parameters through the requirement that the probabilities must sum to $1$. 

Bergsma  and  Rudas (2002a) showed (their Theorem 3) that for two marginals $\mathcal{M}$ and $\mathcal{N}$ and effect $\mathcal{E} \subseteq \mathcal{M} \cap \mathcal{N}$, the partial derivatives of the parameters $\lambda^{\mathcal{N}}_\mathcal{E}$ and $\lambda^{\mathcal{M}}_\mathcal{E}$ according to the components of the probability distribution, evaluated at the uniform distribution, are equal and therefore these parameters cannot be parts of a smooth parameterization of all distributions, because the partial derivative matrix would not always be of full rank. A more detailed analysis of this issue is given by Colombi  and  Forcina (2014), using a different marginal log-linear parameterization which does not involve averaging over the categories of the conditioning variables as in (\ref{MLL}).\footnote{For alternative parameterizations see Section 5.}

One has the following result connecting marginal log-linear parameters of the same effect calculated in different marginals.

\begin{theorem}
Let all the variables be binary, and then each marginal log-linear parameter has one non-redundant value. Let further $\mathcal{E} \subseteq \mathcal{M}\subset \mathcal{N}$. Then
$$
\lambda^{\mathcal{N}}_\mathcal{E}= \lambda^{\mathcal{M}}_\mathcal{E} + f(\boldsymbol{\Lambda}_{\mathcal{N} | \mathcal{M}}),
$$
for some smooth function $f$, with 
$$
\boldsymbol{\Lambda}_{\mathcal{N} | \mathcal{M}}= \{\lambda^{\mathcal{N}}_{\mathcal{F}}: \mathcal{F} \subseteq \mathcal{N}, \mathcal{F} \nsubseteq \mathcal{M}    \}.
$$
Further, 
\begin{equation}\label{indcond}
f(\boldsymbol{\Lambda}_{\mathcal{N} | \mathcal{M}})= 0 \makebox{ if } \big(\mathcal{N} \setminus \mathcal{M}\big) \ind A \,|\, \big(\mathcal{M} \setminus A \big)
\end{equation}
for some $A  \in \mathcal{E}$.

\end{theorem}
\begin{proof}
This is part of Theorem 3.1 in Evans (2015).
\end{proof}

 For example, the second claim of the theorem implies that   if $A \ind B | C$, then $\lambda^{ABC}_B=\lambda^{BC}_B$. This is directly seen by noting that these log-linear parameters are simple functions of the conditional odds of the categories of $B$. For the first one, conditioning is on $A$ and $C$ and for the second one conditioning is on $C$ only. But if the conditional independence in (\ref{indcond}) holds, the conditioning on $A$ does not provide further information after conditioning on $C$ in the sense that
 $$
 P(B=j | C=k)=  P(B=j | A=i, C=k),
 $$
so the conditional probabilities entering the formulas for the log-linear parameters are the same.
In general, this implies that if condition (\ref{indcond}) holds for a distribution, then $\lambda^{\mathcal{M}}_\mathcal{E}$ and $\lambda^{\mathcal{N}}_\mathcal{E}$ cannot be both contained in a smooth parameterization, because then the Jacobian could not be of full rank.

\subsection{ Collapsibility}

 The final property of marginal log-linear parameterizations that we consider before giving the general definition of marginal log-linear models, is collapsibility.
 
 Collapsibility of a parameter is a desirable property but it cannot always be achieved. The concept of collapsibility has many variants, and it refers to the property that some aspect of the inference from a full table is identical to  the corresponding inference based on a marginal table. For example, if, in a $3$-way binary table, $\lambda^{ABC}_{AB}=0$ does not generally imply that $\lambda^{AB}_{AB}=0$, so the inference with respect to the strength of association between variables $A$ and $B$ is not the same, whether it is considered in the full table or in the $AB$ marginal.

 In general, a marginal log-linear parameter $\lambda^{\mathcal{N}}_{\mathcal{E}}$ would be called collapsible (Ghosh  and  Vellaisamy, 2019) if, for $\mathcal{M} \subseteq \mathcal{N} $,  $\lambda^{\mathcal{N}}_{\mathcal{E}} = \lambda^{\mathcal{M}}_{\mathcal{E}}$ held.
 Of course, this cannot be true in general, as in this case marginal log-linear parameters would not be different from the standard log-linear ones. Even for a much weaker requirement, called directional collapsibility, where only the direction of the association is retained, Rudas (2015) showed that there is essentially only one parameterization of multivariate binary distributions which is directionally collapsible for every distribution, and it is not a log-linear, but rather a linear function of the cell probabilities. 
 
 Thus, collapsibility is often interpreted as a property not associated with a parameter, but rather with a parameter and a particular distribution. For example, Ghosh  and  Vellaisamy (2019) gave the following result.
 
 \begin{theorem}
 Let $\emptyset \neq \mathcal{E} \subseteq \mathcal{M} \subsetneq \mathcal{N} \subseteq \mathcal{V}$ be fixed. Then, in the binary case,  collapsibility in the sense that 
 $$
  \lambda^{\mathcal{M}}_{\mathcal{F}} - \lambda^{\mathcal{N}}_{\mathcal{F}} = 0, \makebox{ for all  } \mathcal{F} \subseteq \mathcal{E}
 $$
 holds if and only if for the distribution $P$,
 $$
 \sum_{\mathcal{F} \subseteq \mathcal{E}} \frac{ (-1)^{|\mathcal{E} \setminus \mathcal{F}|}}{2^{|\mathcal{M} \setminus \mathcal{F}|}}  
 \sum_{m: (m)_{\mathcal{F}}=(m^*)_{\mathcal{F}}}  d(\mathcal{M}, m) =0 
 $$
 for all category combinations  $m^*$  of the variables in $ \mathcal{M}$,
 where
 $$
 d(\mathcal{M}, m) =  \log P_{\mathcal{M}}(m) - \frac{1}{2^{|\mathcal{N} \setminus \mathcal{M}|}} \sum_{n: (n)_{\mathcal{M}}=m} \log P_{\mathcal{N}}(n).
 $$
 \end{theorem}
 \begin{proof}
  This is part of Theorem 3.1 in Ghosh  and  Vellaisamy (2019).\footnote{Note that formula (iii) in Theorem 3.1 in Ghosh  and  Vellaisamy (2019) appears to have a typo.}
 \end{proof}

 \section{Marginal log-linear models}
 \label{MLLM}

 Marginal log-linear models are obtained from marginal log-linear parameterizations by applying a linear restriction to the components. If in the example of Section \ref{repmeas}, one wishes to assume that the strength of association between the first and second measurements are the same, that is, treatment does not affect association, then this model may be formulated by requiring that
 $$
 \lambda^{A_1B_1}_{A_1B_1}= \lambda^{A_2B_2}_{A_2B_2}. 
 $$

 
 For example, the graphical model associated with Figure 
 \ref{fig1},  which has been discussed repeatedly, is equivalent to the restrictions in (\ref{ind}) and
 (\ref{cind}) in Section \ref{grm}. Then, in Section \ref{par2}, a parameterization of the joint distribution of the variables based on the marginals 
 $$
 \{AB\}, \,\, \{ABC\}, \,\, \{ABD\}, \,\, \{ABCD\}
 $$
 was considered and  it was shown that the restrictions defining the model may be imposed by restricting some resulting parameters. In Section \ref{def} it was mentioned that the restrictions are the same as 
 \begin{equation}\label{zeros}
 \lambda^{AB}_{AB} = 0 \makebox{ and } \lambda^{ABCD}_{CD} = 0. 
 \end{equation}
 This is the marginal log-linear definition of the graphical model associated with the DAG in Figure 
 \ref{fig1}. Section \ref{graph} will discuss the marginal log-linear approach to graphical modelling in general.
 
 In general, a non-decreasing sequence of marginals is selected and the implied marginal log-linear parameterization is considered. Remember that only non-redundant parameters are included in the parameterization, which is thus smooth, see Theorem \ref{smthness}. In the generality considered in Bergsma  and  Rudas (2002a), a marginal log-linear model is obtained by assuming that the parameters belong to a linear subspace of the parameter space and marginal log-linear models are the special case when the subspace is defined by the equality-to-zero assumptions. 
 
 These models provide a rich family of generalizations of the log-linear model. The actual meaning of the model depends on the marginals selected and on the restrictions applied. Several examples will be discussed later on in the chapter.  
 
 In this section, we concentrate on the general properties of marginal log-linear models.
The first property is that these models always exist.

 \begin{theorem}\label{notempty}
 A marginal log-linear model based on a non-decreasing ordering of the marginals is never empty.
 \end{theorem}
 \begin{proof}
 This is implied directly by Theorem 7 of Bergsma  and  Rudas (2002a).
 \end{proof}
 
 An example is the uniform distribution over a contingency table, which satisfies any marginal log-linear model referred to in the theorem. Note that variation independence is not required here. 
 

The smoothness of the parameterization (see Theorem \ref{smthness}) from which marginal log-linear models are derived implies that the usual desirable asymptotic behaviour holds under Multinomial (see, e.g., Rudas (2018)) sampling.

\begin{theorem}
Assume a marginal log-linear model based on a non-decreasing sequence of marginals contains the true distribution. Then, under Multinomial sampling, the probability that a unique maximum likelihood estimate of the true distribution (or of its parameters) exists tends to $1$ as the sample size goes to infinity. Further, the asymptotic distribution of the maximum likelihood estimator is normal, with expected value equal to the true distribution.
\end{theorem} 
\begin{proof}
This follows from Theorem 8 in Bergsma  and  Rudas (2002a). 
\end{proof} 
This result also implies the standard asymptotic behaviour of goodness-of-fit statistics.

 \section{Alternative parameterizations of marginal log-linear models}\label{altpar}
 
  There are several ways in which odds ratios may be defined and  used to parameterize distributions. These lead to alternative definitions of  marginal log-linear parameterizations and models, adding further flexibility of interpretation to the approach described in this chapter.
 
 It was illustrated in Section \ref{def} that marginal log-linear parameters are closely related to local odds ratios and their higher dimensional generalizations (see, e.g., Rudas (1998, 2018)).   In fact, the marginal log-linear parameters may  be derived from the local $(l-1)$th order odds ratios in the marginal tables. To define these in marginal tables, let the marginal probabilities in the marginal $\mathcal{M}$ be denoted as $P_{\mathcal{M}}$ and let $i_{\mathcal{M} \setminus \mathcal{E}}$ denote a fixed category of the variables in  $\mathcal{M} \setminus \mathcal{E}$. Further, let variable $V_j$ have indices $1, \ldots, c_j$. 
 Then,  the local odds ratio of order $l-1$ in the marginal table for every $ (i_1, \ldots , i_l): \, i_j \ge 2, \, j=1, \ldots, l $, has the form
 \begin{equation}\label{LOR}
\prod_{m_j \in \{0,1 \}, \, j=1, \ldots, l} P^{(-1)^{m_1+ \cdots + m_l}}_{\mathcal{M}} ( i_1-m_1, \ldots i_l-m_l, i_{\mathcal{M} \setminus \mathcal{E}}).
 \end{equation}
The expression in (\ref{LOR}) is a product of probabilities or their reciprocals. The probabilities involved  are in the marginal table $\mathcal{M}$ and are associated with adjacent cells which are obtained by reducing some indices in $ (i_1, \ldots , i_l)$ by $1$. Whether or not (\ref{LOR}) contains a probability or its reciprocal depends on the parity of the number of indices which were reduced.

 Instead of local odds ratios (of any order), spanning cell odds ratios could also be used to define marginal log-linear parameters. For a $2$-way $I \times J$ table, with indices of the variables $0, 1, \ldots, I-1$ and $0, 1, \ldots , J-1$, the spanning cell odds ratios are the odds ratios in the $2 \times 2$ subtables, spanned by the reference cell $(0,0)$ and the spanning cells $(i,j)$, with $i=1, \ldots, I-1$, and $j=1, \ldots, J-1$. The spanning cell odds ratios of order $l-1$ are of the form
 \begin{equation}\label{SPOR}
\prod_{m_j \in \{0,i_j-1 \}, \, j=1, \ldots, l} P^{(-1)^{m}}_{\mathcal{M}} ( i_1-m_1, \ldots i_l-m_l, i_{\mathcal{M} \setminus \mathcal{E}}),
\end{equation}
where $m$ is the number of indices $j$ where $m_j \neq 0$. In this case, the relevant cells are obtained by replacing some indices by $1$.

The intuitive meaning of the higher order odds ratios -- whether local or spanning cell -- is best understood through a recursive definition involving ratios of lower order conditional odds ratios. While local odds ratios measure the strength of association in adjacent cells, and are also relevant when the categories of the variables have orderings, spanning cell odds ratios measure the strength of association when categories are compared to the reference category coded as $1$.

 Bartolucci, Colombi,  and  Forcina (2007) considered various marginal interaction parameters which, if calculated in a non-decreasing set of marginals, may also be used to define  marginal models. These generalized marginal interactions are contrasts of logarithms of sums of (marginal) probabilities. Note that the marginal log-linear parameters considered so far in this chapter are  also contrasts of logarithms of  (marginal) probabilities.
 
 The central concept in the definition of the interaction parameters by Bartolucci et al. (2007) is the lumped table. While local and spanning cell odds ratios derive binary sub-tables from a marginal table by selecting various subsets of the cells, and then calculate the odds ratios for these subsets, the approach of Bartolucci et al. (2007) derives binary sub-tables by collapsing categories of variables. 
 The global and continuation odds ratios resulting from collapsing categories are particularly useful when the variables are ordinal.
 A table formed by the variables with collapsed categories is called a lumped table.
 
 For example, if one considers a bivariate marginal $\mathcal{M}$ with $I \times J$ categories of the variables, and probabilities $P_{\mathcal{M}}(i,j)$, then for each $i^*=2, \ldots I$ and $j^*=2, \ldots J$, one may consider the following quantities:
 $$
 Q_{\mathcal{M},i^*,j^*}(l,l) = \sum_{i=1, \ldots, i^*-1, \, j=1, \ldots, j^*-1 }  P_{\mathcal{M}}(i,j)
 $$
 $$
 Q_{\mathcal{M},i^*,j^*}(l,nl) = \sum_{i=1, \ldots, i^*-1, \, j=j^*, \ldots, J }  P_{\mathcal{M}}(i,j)
 $$
 $$
 Q_{\mathcal{M},i^*,j^*}(m,nl) = \sum_{i=i^*, \ldots, I, \, j=1, \ldots, j^*-1 }  P_{\mathcal{M}}(i,j)
 $$
 $$
 Q_{\mathcal{M},i^*,j^*}(nl,nl) = \sum_{i=i^*, \ldots, I, \, j=j^*, \ldots, J }  P_{\mathcal{M}}(i,j).
 $$
 Here, the summation of the marginal cell probabilities goes for the indices less ($l$) or not less ($nl$) than the specified $i^*$ and $j^*$.
 
 Then, the lumped table is of the size $2 \times 2 $, and the lumped distribution is $Q_{\mathcal{M},i^*,j^*}$. This kind of lumping divides the cells of the marginal table into $4$ rectangles and combines the probabilities within each. The odds ratio of the lumped distribution is
 $$ 
 \frac{Q_{\mathcal{M},i^*,j^*}(l,l)Q_{\mathcal{M},i^*,j^*}(nl,nl)}{Q_{\mathcal{M},i^*,j^*}(l,nl)Q_{\mathcal{M}}(nl,l)},
 $$
 which is called the global odds ratio belonging to cell $(i^*,j^*)$. Similar lumping is also possible for $l$-dimensional tables, and the $(l-1)$th order odds ratio in the resulting $2^l$ table is also called a global odds ratio. There are $(c_1-1)(c_2-1) \cdots (c_l-1)$ global odds ratios for an effect $\mathcal{E}$.
 
 Another type of odds ratio is obtained by the following partial lumping for $2$-way $I \times J$ tables, for each $i^*=1, \ldots I-1$, and $j^*=1, \ldots J-1$:
 $$
 R_{\mathcal{M},i^*,j^*}(e,e) =   P_{\mathcal{M},i^*,j^*}(i^*,j^*)
 $$
 $$
 R_{\mathcal{M},i^*,j^*}(n,e) =   P_{\mathcal{M}}(i^*+1,j^*)
 $$
 $$
 R_{\mathcal{M},i^*,j^*}(e,m) = \sum_{ j=j^*+1, \ldots, J}  P_{\mathcal{M}}(i^*,j)
 $$
 $$
 R_{\mathcal{M},i^*,j^*}(n,m) = \sum_{j=j^*+1, \ldots, J }  P_{\mathcal{M}}(i^*+1,j),
 $$
 where $e$ stands for equal, $n$ stands for next, and $m$ stands for more than.

 The odds ratio obtained for the lumped $2 \times 2$ distribution,
  $$ 
 \frac{R_{\mathcal{M},i^*,j^*}(e,e)R_{\mathcal{M},i^*,j^*}(n,m)}{R_{\mathcal{M},i^*,j^*}(e,m)R_{\mathcal{M},i^*,j^*}(n,e)},
 $$
 is called the continuation odds ratio. Its meaning is best seen by writing it as
 $$
 \frac{R_{\mathcal{M},i^*,j^*}(n,m)/R_{\mathcal{M},i^*,j^*}(n,e)}{R_{\mathcal{M},i^*,j^*}(e,m)/R_{\mathcal{M},i^*,j^*}(e,e)},
 $$
 which is the ratio of the conditional odds of the 'continuation' of the second variable, as opposed to not changing it, when conditioned on the next or on the current category of the first variable.

 In multivariate generalizations of the continuation odds ratios, lumping does occur for the response variables but not for the explanatory variables, if such a distinction among the variables exists.

 Bartolucci et al. (2007) define extended interaction parameters as contrasts of logarithms of generalized odds ratios including global and continuation odds ratios (and also local and spanning cell odds ratios) and show that for the models obtained by linear restrictions on these, many of the results presented so far in this chapter apply, too. They called this more general model class {\it hierarchical marginal models}.

\section{Marginal log-linear parameterization of conditional independence models}

Many of the relevant marginal models assume conditional independences in various marginals of the table. The most important group of such models are graphical models, of which models associated with DAGs have already been considered. A more detailed account will be given later in this chapter. In this section, we present general results about formulating conditional independence models as marginal models, that is, by restricting some parameters in a hierarchical and complete marginal log-linear parameterization.

For $i=1, \ldots, h$, let $\mathcal{A}_i \neq \emptyset$, $\mathcal{B}_i\neq \emptyset$, and $\mathcal{C}_i$ be pairwise disjoint sets. The goal is to formulate the following conditional independences jointly, as a marginal log-linear model:
\begin{equation}\label{coins}
\mathcal{A}_i \ind \mathcal{B}_i \, | \, \mathcal{C}_i, \makebox{ for all } i = 1, \ldots, h. 
\end{equation}
For example, the graphical model associated with the DAG in Figure \ref{fig1} is equivalent to imposing the conditional independences (\ref{ind}) and (\ref{cind}). In this case, $h=2$, $\mathcal{A}_1= \{A\}$, $\mathcal{B}_1=\{B\}$, $\mathcal{C}_1= \emptyset$ and $\mathcal{A}_2= \{C\}$, $\mathcal{B}_2=\{D\}$, $\mathcal{C}_2= \{A,B\}$.

To explore when (\ref{coins}) may be formulated as a marginal log-linear model, define
$$
\mathbb{D}_i= \mathbb{P}(\mathcal{A}_i \cup \mathcal{B}_i \cup \mathcal{C}_i) \setminus [\mathbb{P}(\mathcal{A}_i \cup  \mathcal{C}_i) \cup
\mathbb{P}( \mathcal{B}_i \cup \mathcal{C}_i)],
$$
where $\mathbb{P}(.)$ denotes the power set. That is, for every $i$, $\mathbb{D}_i$ is the collection of those subsets of $\mathcal{A}_i \cup \mathcal{B}_i \cup \mathcal{C}_i$ that contain variables from both $\mathcal{A}_i$ and $\mathcal{B}_i$. In the case of the DAG example, $\mathbb{D}_2= \{CD, ACD, BCD, ABCD \}$. A sufficient condition is given by the following result.
\begin{theorem}\label{d-th}
Let 
$$
\mathcal{M}_1, \ldots, \mathcal{M}_k =\mathcal{V}
$$
be a non-decreasing sequence of marginals with the following property:
\begin{equation}\label{d-cond}
\mathcal{C}_i \subseteq \mathcal{M}(\mathcal{E}) \subseteq \mathcal{A}_i \cup \mathcal{B}_i \cup \mathcal{C}_i, \makebox{ for all } \mathcal{E} \in \cup_{i=1}^h\mathbb{D}_i. \end{equation}

Then, the conditional independences in (\ref{coins}) define a marginal log-linear model based on these marginals. More specifically, (\ref{coins})  holds for a distribution $P$ if and only if
$$
\lambda^{\mathcal{M}(\mathcal{E})}_{\mathcal{E}} = 0 \makebox{ for all } \mathcal{E} \in \cup_{i=1}^h\mathbb{D}_i
$$
for this distribution. Further, the distributions in the model are smoothly parameterized by the remaining marginal log-linear parameters:
$$
\{\lambda^{\mathcal{M}(\mathcal{E})}_{\mathcal{E}}:  \mathcal{E} \notin \cup_{i=1}^l\mathbb{D}_i\}.
$$
\end{theorem}
\begin{proof}
This is part of Theorem 1 in Rudas et al. (2010).
\end{proof}

Condition (\ref{d-cond}) means that for any effect $\mathcal{E}$ which contains variables from any two subsets $\mathcal{A}_i$ and $\mathcal{B}_i$ of variables which are assumed to be conditionally independent, the first marginal in the sequence which contains $\mathcal{E}$ has to be big enough to contain the conditioning set $\mathcal{C}_i$, but has to be small enough to be contained in $\mathcal{A}_i \cup \mathcal{B}_i \cup \mathcal{C}_i$. The condition requires the sequence of marginals to be sufficiently rich.

For example, in the case of the model defined by the DAG in Figure 1, there are various choices of the sequence of marginals with property (\ref{d-cond}). Clearly, the $AB$, $ABCD$ sequence is one such choice. But $A$, $AB$ $ABCD$ or $A$, $B$, $AB$, $ABC$, $ABD$, and $ABCD$ are also  appropriate sequences of marginals in order to be able to define the model by setting some marginal log-linear parameters to zero. However, Theorem \ref{d-th} does not imply that the DAG model would be a marginal log-linear model based on the sequence of marginals $A$, $ACD$, $ABC$, $ABCD$, because the variables $C$ and $D$, which have to be conditionally independent, are present together, without their conditioning set $AB$. In the case of the sequence of marginals $A$, $AB$,  and $ABCD$, the effects which are to be set to zero to specify the DAG model are
$$
AB, \, \, CD, \, \, ACD, \, \, BCD, \, \, ABCD;
$$
the first one in the $AB$ marginal, and the others in the $ABCD$ marginal. Note that this is the same specification as the one given in (\ref{zeros}), taking into account that the marginal log-linear parameters are log-linear parameters calculated in a marginal, thus the second equality in (\ref{zeros}) implies that
$$
\lambda^{ABCD}_{ACD}=\lambda^{ABCD}_{BCD}=\lambda^{ABCD}_{ABCD}=0;
$$
see Rudas (2018).

The marginal log-linear parameters which parameterize the distributions in the DAG model belong to the following effects:
$$
\emptyset, \,\, A, \,\, B, \,\,\ C, \,\, D, \,\, AC, \,\, BC, \,\,\ ABC, \,\, AD, \,\, BD, \,\, ABD.
$$
Out of these, the first two are calculated in the first marginal, the third one in the second marginal, and the rest in the last marginal. 

Further applications of Theorem \ref{d-th} will be given in Section \ref{graph}.

The result in Theorem \ref{d-th} raises the question of how to determine, for a given list of conditional independences, whether a smooth marginal log-linear definition and parameterization of the model is possible. This would require considering the non-decreasing sequences of marginals in which the required conditional independences in (\ref{coins}) may be formulated, and to see whether (\ref{d-cond}) holds for any such sequence. An apparent difficulty is that a particular effect $\mathcal{E}$ may be a subset of $\mathbb{D}_i$ for more than one $i$ and, thus, (\ref{d-cond}) may impose several restrictions on $\mathcal{M}(\mathcal{E})$. 
An obvious necessary condition for the existence of a smooth marginal log-linear definition is the following.
If for a subset of the variables $\mathcal{E}$, $I_{\mathcal{E}}$ denotes the indices from among $1, \ldots, h$, for which $\mathcal{E} \subseteq \mathbb{D}_i$, then $\mathcal{M}(\mathcal{E})$ should be such that
$$
\cup_{i \in I_{\mathcal{E}} } \mathcal{C}_i \subseteq \mathcal{M}(\mathcal{E}) \subseteq \cap_{i \in I_{\mathcal{E}}} \mathcal{A}_i \cup \mathcal{B}_i  \cup \mathcal{C}_i   
$$
and if
$$
\cup_{i \in I_{\mathcal{E}} } \mathcal{C}_i  \nsubseteq  \cap_{i \in I_{\mathcal{E}}} \mathcal{A}_i \cup \mathcal{B}_i  \cup \mathcal{C}_i, 
$$
then, an appropriate sequence of marginals does not exist and the sufficient condition of Theorem \ref{d-th} must hold.

Forcina et al. (2010) arrived at similar results using a different approach. They proposed an algorithm to decide whether a model defined by a given set of conditional independences admits a marginal log-linear definition in the sense discussed here, and is, thus, smooth.\footnote{A model is called smooth if it admits a smooth parameterization.}

Forcina (2012) discussed further questions related to the smoothness of models defined by various, more general, collections of conditional independence statements, when, for any ordering of the relevant marginals, the marginal log-linear parameters which have to be set to zero in order to obtain the prescribed conditional independences, cannot be specified in the first marginal where the effect occurs. An example discussed in Forcina (2012) is for four binary variables and requires that
$$
X_1 \ind X_2 \, | \, X_3,
$$
$$
X_2 \ind X_3 \, | \, X_4,
$$
$$
X_2 \ind X_4 \, | \, X_1.
$$
He showed the model is smooth, even though condition (\ref{d-cond}) does not hold, and Theorem \ref{d-th} does not apply. Indeed, if, for example, the three marginals appeared in the following order,
$$
X_1X_2X_3, \, X_2X_3X_4, \, X_1X_2X_4, 
$$
then $X_4 \nsubseteq \mathcal{M}(X_2X_3)=\{X_1X_2X_3\}$. Similarly, all other orderings of the marginals would lead to a violation of (\ref{d-cond}).

Forcina (2012) offered an iterative algorithm to construct the distributions in the model based on the mixed parameterization of exponential families, see, e.g. Rudas (2018), and proved that the convergence of the algorithm implies smoothness of the model.

\section{Estimation and testing} 

In this section, we describe the maximum likelihood (ML) and the GEE approaches to estimating marginal log-linear models. Both estimation methods provide asymptotically unbiased estimators, but ML estimators have the advantage over GEE ones of being asymptotically efficient. On the other hand, the GEE approach has the advantage of being computationally more efficient, which is important because of the large possible sizes of contingency tables. For example, 8 variables with 5 categories gives a contingency table of size $5^8=390,625$.
The ML method requires all expected cell frequencies to be estimated, whereas the GEE method only estimates first and second moments of observed marginal frequencies. Nevertheless, the ML method can handle large tables; for example, we found that tables with one million cells can be estimated without too much difficulty. 

ML estimators of marginal log-linear models are, in general, not available in closed form and iterative methods need to be used. 
There are two main approaches. Firstly, there are algorithms based on the approach developed by Aitchison and Silvey (1959), who used the Lagrange multiplier technique. Lang and Agresti (1994) first used this method for marginal models, and a modification, which seems to have improved practical performance, was given by Bergsma (1997). A difficulty with these methods is that a search is done for a saddle point, hence convergence may be difficult to monitor. Bergsma and Rapcsak (2006) resolved this problem by developing an alternative Lagrangian method, which turns the constrained maximization problem into an unconstrained one.

A second approach to ML estimation is to  maximize the likelihood parameterized in terms of a hierarchical and complete marginal log-linear parameter vector, for example, using a Fisher scoring algorithm. 
The drawback of this approach is that it involves `iteration within iteration', that is, at each Fisher scoring step, the cell probabilities need to be computed from the current estimate marginal log-linear parameter (this can be done with the iterative proportional fitting algorithm, which has guaranteed convergence). Therefore, this approach is computationally burdensome and Lagrange multiplier methods are more attractive. We describe the approach for completeness.

As far as we are aware, the GEE method has not been described in the literature for general marginal models. Section~\ref{sec-gee} gives an outline, including a suitable choice of the working covariance matrix.



\subsection{Matrix formulation of marginal models}

Let $\bmm$ be a vector containing the expected cell frequencies in a contingency table. 
A marginal log-linear parameter $\blambda$ can be represented as
\begin{align} \label{llpar}
    \blambda = \bB'\log\bM'\bmm
\end{align}   
where $\bB$ and $\bM$ are appropriately defined matrices and a prime represents the transpose. 
This formulation includes the marginal log-linear parameterizations of Bergsma and Rudas (2002), see Section~\ref{MLLP}.

A marginal log-linear model is then defined by
\begin{align}\label{freedom}   \blambda = \bX\bbeta  \end{align}
for a matrix $\bX$ and parameter vector $\bbeta$ of smaller length than $\blambda$. 
Equivalently, a marginal log-linear model can be specified as
\begin{align}\label{constraint}   \bC'\blambda = \bzero  \end{align}
for an appropriate matrix $\bC$. Taking $\bC$ to be the orthogonal complement of $\bX$, in the sense that $\bC'\bX=\bzero$ and $(\bX,\bC)$ is an invertible matrix, the two formulations are seen to be equivalent.
These formulations have been called {\it freedom} and {\it constraint} specifications (Lang, 1996a)

For example, consider a $2\times 2$ table with expected cell frequencies $(m_{11},m_{12},m_{21},m_{22})$. The marginal homogeneity model in the constraint specification is $m_{i+}=m_{+i}$ ($i=1,2$), where a plus in the subscript denotes summation over that subscript. In matrix notation, this is 
\begin{align}\label{constr2x2}   
    \begin{pmatrix} 1  & 0 & -1 & 0 \\ 0 & 1 & 0  & - 1  \end{pmatrix}
    \log\left[
    \begin{pmatrix} 1 & 1 & 0 & 0 \\ 0 & 0 & 1 & 1  \\ 1 & 0 & 1 & 0 \\ 0 & 1 & 0 & 1  \end{pmatrix}
    \begin{pmatrix} m_{11} \\ m_{12} \\ m_{21} \\ m_{22}  \end{pmatrix}\right]  = 0  .
\end{align}
In the freedom specification, the model is $(m_{i+},m_{+i})=(\beta_i,\beta_i)$ ($i=1,2$), which in matrix notation is
\[   
    \log\left[
    \begin{pmatrix} 1 & 1 & 0 & 0 \\ 0 & 0 & 1 & 1  \\ 1 & 0 & 1 & 0 \\ 0 & 1 & 0 & 1  \end{pmatrix}
    \begin{pmatrix} m_{11} \\ m_{12} \\ m_{21} \\ m_{22}  \end{pmatrix}
    \right]  = 
    \begin{pmatrix} 1 & 0 \\ 0 & 1  \\   1 & 0 \\ 0 & 1  \end{pmatrix}
    \begin{pmatrix} \beta_1 \\ \beta_2 \end{pmatrix}  .
\]

\subsection{Characterization of ML estimators}

In this section we give a score equation, and a Lagrangian score equation, whose solutions, under some conditions, are the ML estimators of a marginal model. Algorithms for solving these equations are postponed to Section~\ref{sec-alg}.

Let $\bn$ be a vector of observed cell counts of a contingency table. 
We assume $\bn$ has a multinomial or independent Poisson distribution with expected frequency vector $\bmm=E(\bn)$. The log-likelihood for $\bmm$ then is
\begin{align}\label{loglik}  L(\bmm|\bn) = \bn'\log(\bmm) - \bone'\bmm + c \end{align}
where $\bone$ is a vector of ones of appropriate length and $c$ is a constant. 
In the multinomial case, the constraint $\bone'\bmm=\bone'\bn$ holds, but this does not affect maximum likelihood estimation or inference in the present case (Lang, 1996b). Hence, for notational simplicity, we will ignore the multinomial constraint below. 
The maximum likelihood estimator $\hat\bmm$ of $\bmm$ under a marginal log-linear model maximizes the log-likelihood $L(\bmm|\bn)$ subject to a constraint of the form~(\ref{freedom}) or~(\ref{constraint}). 
The maximum likelihood estimator $\hat\bmm$ of $\bmm$ has been characterized in two equivalent ways, namely as the solution to (i)  equations involving Lagrange multipliers, or (ii) the score equation for $\bbeta$. The former is due to Aitchison and Silvey (1958) and Lang (1996a) and the latter was considered by Glonek and McCullagh (1995) and Colombi and Forcina (2001).

The Lagrange multiplier method seeks a stationary point of the Lagrangian log-likelihood
\[  L(\bmm,\btau|\bn) = \bn'\log(\bmm) - \bone'\bmm - \btau'\bC'\blambda \]
where $\btau$ is a vector of Lagrange multipliers and $\blambda$ is a marginal log-linear parameter of the form~(\ref{llpar}). 
Denote the Jacobian of $\blambda$ as $\bLambda$, given by
\begin{align}\label{Lambdadef}   \bLambda = \frac{d\blambda'}{d\bmm} =  \bM\bD_{\bM'\bmm}^{-1}\bB \end{align}
where $\bD$ is the diagonal matrix with its subscript on the main diagonal. 
Differentiating the log-likelihood $L$ with respect to $\bmm$ and equating to zero gives 
\begin{align} \label{lagr}
    \frac{\bn}{\bmm} - \bone + \bLambda\bC\btau = \bzero 
\end{align}   
where the division in $\bn/\bmm$ is element-wise. 

Under some conditions, the ML estimator $\hat\bmm$ is a solution to the simultaneous equations~(\ref{lagr}) and~(\ref{constraint}). 
Sufficient conditions include (i) all observed frequencies are strictly positive, and (ii) the Jacobian $\bLambda\bC$ has full column rank. For most, if not all, marginal models of practical interest, the second condition is satisfied; see Section~\ref{sec-smooth}. However, the positivity of all observed frequencies is often not satisfied in practice; for example, for many real-world problems the number of cells in the table is larger than the sample size, implying there must be some cells with zero observations. 
A heuristic solution to this problem is to replace all zero observed frequencies by a small constant, so that the total contribution to the likelihood will be negligible (Bergsma, Croon  and  Hagenaars, 2009).

To illustrate the problem with zero observed cells, note that for the marginal homogeneity model defined by~(\ref{constr2x2}),~(\ref{lagr}) becomes
\[   \frac{n_{ij}}{m_{ij}} - 1 - \frac{\lambda_i}{m_{i+}} + \frac{\lambda_j}{m_{+j}} = 0  \quad i,j=1,2. \]
Consider now the equation for $(i,j)=(1,1)$. Since $m_{1+}=m_{+1}$, we obtain
\[   \frac{n_{11}}{m_{11}} - 1  = 0 .  \]
The solution is $\hat m_{11}=n_{11}$ except if $n_{11}=0$, in which case there is no solution. The true ML estimator in this case is $\hat m_{11}=0$, and replacing $n_{11}$ by a small number makes negligible difference for inferential purposes.

An alternative to the Lagrange multiplier method for characterizing the ML estimator is by means of the score equation for $\bbeta$ in~(\ref{freedom}). The likelihood is parameterized in terms of $\bbeta$ and 
the ML estimator is obtained by computing the score equation and solving for $\bbeta$. 
This approach is facilitated if $\blambda$ is a marginal log-linear parameterization, in which case its Jacobian $\bLambda$ is invertible. Differentiating the log-likelihood then gives the score vector
\begin{align}\label{score}  \bs(\bbeta) := \frac{d L}{d\bbeta}  =  \frac{d\blambda'}{d\bbeta}\, \frac{d\bmm}{d\blambda'}\, \frac{dL}{d\bmm}  =  \frac{d\blambda'}{d\bbeta}\, \Big(\frac{d\blambda'}{d\bmm}\Big)^{-1}\, \frac{dL}{d\bmm} = \bX'\bLambda^{-1}\Big( \frac{\bn}{\bmm} - \bone \Big).
\end{align}
Provided all observed cell frequencies are positive, the ML estimator $\hat\bbeta$ satisfies $\bs(\hat\bbeta)=0$. As in the Lagrange multiplier case, we suggest replacing zero observed cell frequencies by a small constant. 
If $\blambda$ is a smooth parameterization, then $\bLambda$ is invertible, and $\bs(\bbeta)=0$ is equivalent to the Lagrangian equation~(\ref{lagr}), since $\bC'\bX=\bzero$ and $(\bX,\bC)$ is an invertible matrix.
The score function is potentially computationally expensive to evaluate, because the matrix $\bLambda$ needs to be computed and inverted.

\subsection{Likelihood ratio tests and asymptotic distribution of ML estimators}

Suppose model~(\ref{freedom}) holds. Following results of Aitchison and Silvey (1958) and Lang (1996a), the maximum likelihood estimator $\hat\bmm$ under this model has an approximate large sample multivariate normal distribution, with mean $\bmm$ and covariance matrix
\[  \cov(\hat\bmm) \approx \bD_\bmm - \bLambda(\bLambda'\bD_\bmm\bLambda)^{-1}\bLambda   .  \]
The estimated parameter 
vector $\hat\bbeta=(\bX'\bX)^{-1}\bX'\log\hat\bmm$ also has a large sample multivariate normal distribution, with mean $\bbeta$ and covariance matrix
\[  \cov(\hat\bbeta) = (\bX'\bX)^{-1}\bX'\bD_\bmm^{-1}\cov(\hat\bmm-\bmm)\bD_\bmm^{-1}\bX(\bX'\bX)^{-1}  .  \] 

The usual likelihood ratio test can be used for selecting nested models. Let $H_0$ and $H_1$ be nested models, i.e., if $H_0$ is true then $H_1$ is true, and let $\hat\bmm_k$ be the ML estimate of $\bmm$ under $H_k$ ($k=0,1$). 
The log likelihood ratio test statistic is
\[   G^2 = 2\bn'\log\frac{\hat\bmm_0}{\hat\bmm_1}  \]
Under some regularity conditions, if $H_0$ is true then $G^2$ has an asymptotic chi-square distribution with degrees of freedom (df) equal to the dimension of $H_1$ minus the dimension of~$H_0$. 

Non-nested models can be compared using various information criteria, such as the Bayesion information criterion (BIC),
\[   \text{BIC} = G^2 + 2\,\text{df}\,\log(N) \]
where $N$ is the sample size.


\subsection{Algorithms for finding ML estimators}\label{sec-alg}

\subsubsection{Lagrangian methods}

Several Lagrangian algorithms have been proposed to find the ML estimators of a marginal model. 
In a seminal paper, Aitchison and Silvey (1958) described Lagrangian methods for constrained maximum likelihood in some generality. Lang and Agresti (1994) and Lang (1996a) introduced Langrangian methods for categorical marginal models.
The algorithm we describe here is a slightly modified algorithm developed by Bergsma (1997), which practical experience indicated has improved convergence properties compared to the original Aitchison and Silvey algorithm. 

The first step of the algorithm is to choose an appropriate
starting point $\bmm^{(0)}$, after which subsequent estimates $\bmm^{(k+1)}$ ($k=0,1,2,\ldots$) are calculated iteratively using the
formula
\begin{align}\label{alg1}
   \log\bmm^{(k+1)} = \log \bmm^{(k)} + \step^{(k)}\bu\big(\bmm^{(k)}\big)
\end{align}
where $\step^{(k)}$ is an appropriately chosen step size and
\[
   \bu(\bmm) = \frac{\bn}{\bmm} - \bone - \bLambda\bC\big(\bC'\bLambda'\bD_\bmm\bLambda\bC\big)^{-1}\big[\bC'\bLambda'\bM'(\bn-\bmm)+\bC'\blambda\big]
   .
\] 
Here, $\bLambda$ is defined by~(\ref{Lambdadef}) and depends on $\bmm$. 
A suggested starting point is $\bmm^{(0)}=\bn+\epsilon$, where $\epsilon$ is some small constant, such as $10^{-6}$. 
For further details, see Bergsma, Croon and Hagenaars (2009, Section 2.3.5). 

A closely related algorithm was given by Colombi and Forcina (2001), which, being based on updating $\bbeta$ in~(\ref{freedom}), was named the `regression algorithm'. The two algorithms were shown to be equivalent by Evans and Forcina (2011). 
They showed the two algorithms have rather different numerical properties depending on whether 
the design matrix $\bX$ has a block diagonal structure, arising with continuous covariates: if this is the case, the regression algorithm tends to be much more efficient but if not, Bergsma's algorithm tends to be much more efficient in practice.


Although we have very good practical experience with convergence of the algorithm~(\ref{alg1}) to the ML estimator, theoretical results are lacking. 
Generally speaking, convergence properties of constrained optimization problems are more difficult to establish than those of unconstrained ones. 
The Lagrange multiplier method turns a constrained optimization problem into the problem of finding a saddle point of the Lagrangian function, but finding such a saddle point may be more difficult than finding a global (unconstrained) maximum or a minimum. 
Two ways of reformulating the ML estimation problem for marginal models as an unconstrained optimization problem have been described.

Bergsma and Rapcs\'ak (2006) provided a general method for turning a constrained optimization problem into an unconstrained one 
and applied this to ML estimation of marginal models. The advantage of this algorithm is good theoretical properties, and it is similar in computational efficiency to the algorithm defined by~(\ref{alg1}).

\subsubsection{Fisher scoring}

In this section we build on the Fisher scoring algorithm for marginal models described by Colombi and Forcina (2001).
We wish to find the value of $\bbeta$ in~(\ref{freedom}) maximizing the log-likelihood~(\ref{loglik}). Here, $\blambda$ is a marginal log-linear parameterization as described in Section 3.
Then, by Theorem~\ref{smthness}, $\bLambda$ defined by~(\ref{Lambdadef}) is invertible.
Differentiating the log-likelihood gives the score vector $\bs(\bbeta)$ given by~(\ref{score}). 
The Fisher information on $\bbeta$ is
\[  \bI(\bbeta) := - E\Big[ \frac{d\bs(\bbeta)}{d\bbeta'}\Big]  =  \bX'\bLambda^{-1}\bD_\bmm^{-1}\bLambda'^{-1}\bX.  \]
The Fisher scoring algorithm is given by
\begin{align} \label{fish1}   \bbeta^{(k+1)} = \bbeta^{(k)} + \text{step}^{(k)}\bI(\bbeta^{(k)})^{-1}\bs(\bbeta^{(k)}). \end{align}
At each iteration, the vector of expected cell frequencies $\bmm$ needs to be computed from $\blambda$, which can be done using the iterative proportional fitting algorithm (Bergsma and Rudas, 2002). However, the Newton-Raphson scheme proposed by Glonek and McCullagh (1995) may be numerically more efficient.

A major potential numerical bottleneck for~(\ref{fish1}) is that $\bLambda$ needs to be stored and inverted at each iteration. In particular, if there are $K$ cells in the table $\bLambda$ is a $K\times K$ matrix. 
A normally much more efficient algorithm can be obtained by updating $\blambda=\bX\bbeta$ directly. We obtain the updating step
\begin{align*}
	\blambda 
	&\to \blambda + \step\, \bX\cdot \bI(\bbeta)^{-1}\bs(\bbeta)  \\
	&= \blambda - \bX(\bX'\bLambda^{-1}\bD_\bmm^{-1}\bLambda'^{-1}\bX)^{-1}\bX'\bLambda^{-1}\bX'\bLambda^{-1}\bD_\bmm^{-1}( \bn - \bmm )  \\
	&= \blambda - \big[ \bLambda'\bD_\bmm\bLambda - \bLambda'\bD_\bmm\bLambda \bC(\bC'\bLambda'\bD_\bmm\bLambda\bC)^{-1}\bC'\bLambda'\bD_\bmm\bLambda \big] \bD_\bmm^{-1}( \bn - \bmm ) , 
\end{align*}
where $\bC$ is an orthogonal complement of $\bX$ (see~(\ref{constraint})). 
In practice, the matrix $\bC$ typically has low column rank, making the latter updating step relatively efficient if implemented well; see Colombi and Forcina for details. 


Overall, the Fisher scoring algorithm appears more cumbersome to implement than Lagrangian algorithms, in particular if numerical efficiency is desired. Furthermore, due to the required `iteration within iteration', Fisher scoring algorithms can be expected to be slower than Lagrangian algorithms. If parameterizations based on a set of marginals which is not ordered decomposable are used, out-of range estimates (negative probabilities) can be obtained (Colombi and Forcina, 2001). 


In more general settings, a drawback of Fisher scoring is that it requires a parameterization of the distribution in terms of parameters of interest. Such parameterizations are available for marginal log-linear models, but not for the more general models based on non-log-linear parameters considered by Bergsma (1997), Lang (2005), and Bergsma, Croon and Hagenaars (2009).

\subsubsection{Software}

The following three \texttt{R} packages are available for marginal modelling: \texttt{cmm} by Wicher Bergsma and Andries van der Ark,  \texttt{mph.fit} by Joseph Lang, and \texttt{hmmm} by Roberto Colombi, Sabrina Giordano, and Manuela Cazzaro. A detailed description of the \texttt{cmm} package can be found at \url{stats.lse.ac.uk/bergsma/cmm/index.html}. The website contains R code with explanations for all the data examples in Bergsma et al.\ (2009). 
Documentation for \texttt{mph.fit} can be found at \url{homepage.stat.uiowa.edu/~jblang/mph.fitting/index.htm} and for \texttt{hmmm} at \url{rdrr.io/cran/hmmm/}; see also Colombi, Giordano,  and  Cazzaro (2014).
All three packages can estimate a wide variety of models.
A special feature of \texttt{cmm} is that it can handle marginal models with latent variables, while \texttt{hmmm} can handle hidden Markov models and inequality constraints. For features of \texttt{hmmm},   see also   Section \ref{altpar}.


\subsection{The GEE method}\label{sec-gee}

A drawback of ML estimation of marginal models is that all cells in the contingency table need to be estimated, making it  computationally infeasible if the number of cells is large. The GEE method is a quasi-likelihood method which models the covariance matrix between marginal observations, while ignoring higher order associations, allowing greater computational efficiency at the cost of some statistical efficiency. A detailed general overview of the GEE methodology is provided by Molenberghs and Verbeke (2005, Chapter 8). In most literature on GEE, the association is modelled using correlations. Lipsitz, Laird, and Harrington (1991) developed the GEE methodology based on odds ratios for univariate binary responses. Touloumis, Agresti, and Kateri (2013) gave a more general development for multinomial responses. Below, we adapt the GEE procedure for general marginal models as described in this paper, i.e., the association is modelled using log-linear parameters and the marginals of interest may be multivariate. 


The GEE method derives from the score vector for a generalized linear model for a multivariate marginal mean; if $\by\sim\text{MVN}(\bmu,\bV_\by)$ and $\bmu=g(\bX\bbeta)$ for some link function $\bbeta$, the score equation yielding the maximum likelihood estimator of $\bbeta$ is 
\begin{align} \label{normalscore}
   \frac{d\bmu'}{d\bbeta}\bV_\by^{-1}(\by-\bmu) = \bzero.
\end{align}
This equation can also yield a consistent estimator of $\bbeta$ if $\by$ is non-normal (Wedderburn, 1974). However, there is the difficulty that $\bV_\by$ is typically unknown and potentially difficult to estimate. 
Liang and Zeger (1986) proposed replacing $\bV_\by$ with a potentially incorrect `working' covariance matrix $\tilde\bV_\by$, giving the GEE
\begin{align} \label{gee0}
   \frac{d\bmu'}{d\bbeta}\tilde\bV_\by^{-1}(\by-\bmu) = \bzero.
\end{align}
Here, $\tilde\bV_\by$ can depend on parameters, in particular $\bmu$ and parameters describing the correlation structure of $\by$. 
Liang and Zeger showed that under some conditions, the GEE yields a consistent estimator $\tilde\bmu$ of $\bmu$. Huber's (1967) large sample sandwich estimator of the covariance matrix of $\tilde\bmu$ is then
\begin{align*} 
   \tilde\cov(\tilde\bbeta) = \tilde\bI^{-1}\tilde\bJ\tilde\bI^{-1}
\end{align*}
where 
\begin{align*} 
   \tilde\bI = \frac{d\bmu'}{d\bbeta}\tilde\bV_\by^{-1}\frac{d\bmu}{d\bbeta'} \Big|_{\bbeta=\tilde\bbeta}
   \quad\quad 
   \tilde\bJ = \frac{d\bmu'}{d\bbeta}\tilde\bV_\by^{-1}\bV_\by^*\tilde\bV_\by^{-1}\frac{d\bmu}{d\bbeta'} \Big|_{\bbeta=\tilde\bbeta}.
\end{align*}
Here, $\bV_\by^*$ is a consistent estimator of $\bV_\by$.

Let us now give the GEE method for estimating $\bbeta$ in the marginal model~(\ref{freedom}), denoting the marginal observed frequency vector by $\by=\bM'\bn$ and the corresponding expected frequency vector by $\bmu=E(\by)=\bM'\bmm$. Then
\begin{align}\label{vy} \bV_\by=\bM'\bD_\bmm\bM - N^{-1}\bmu\bmu' \end{align}
where $N=\bone'\bn$ is the sample size.
We can write the marginal model~(\ref{freedom}), with $\blambda$ given by~(\ref{llpar}), as
\begin{align*}
    \bmu = \exp(\bU\bX\bbeta)
\end{align*}
where $\bU$ is an orthogonal complement of $\bB$, that is, $\bB'\bU=\bzero$ and $(\bB,\bU)$ is an invertible matrix. 
Hence,
\[  \frac{d\bmu'}{d\bbeta} = \bX'\bU'\bD_\bmu  \]
so that~(\ref{normalscore}) becomes
\begin{align} \label{normalscore1}
   \bX'\bU'\bD_\bmu\tilde\bV_\by^{-1}(\by-\bmu) = \bzero.
\end{align}

A difficulty is that $\bV_\by$ is typically not invertible, in which case we can replace~(\ref{normalscore1}) by
\begin{align} \label{normalscore2}
    \by - \bmu  + \bV_\by\bD_\bmu^{-1}\bB\bC\btau = \bzero 
\end{align}
where $\btau$ is a parameter to be estimated. 
Straightforward calculations show that if $\bV_\by$ is invertible,~(\ref{normalscore1}) and~(\ref{normalscore2}) are equivalent.
Note that~(\ref{normalscore2}) follows from the Lagrangian score equation~(\ref{lagr}) by pre-multiplying the left- and right-hand sides by $\bM'\bD_\bmm$. Replacing $\bV_\by$ in~(\ref{normalscore2}) by a working covariance matrix $\tilde\bV_\by$ gives a GEE for marginal models. 
A consistent estimator $\bV_\by^*$ of $\bV_\by$ is needed to compute $\tilde\bJ$, and for this we can take $\bV_\by^*=\bM'\bD_\bn\bM$.

It remains to find a working covariance $\tilde\bV_\by$. A simple way to do this is as follows. 
Suppose the marginal model is based on non-nested marginals $\cM_1,\ldots,\cM_k$. Then $\bV_\by$ given by~(\ref{vy}) is a function of the expected marginal frequencies for the following marginals
\begin{align}\label{margij}   \{\cM_i\cup\cM_j|i,j=1,\ldots,k\}. \end{align}
A simple choice of working covariance matrix is obtained by assuming a (potentially incorrect) conditional independence model for the marginal $\cM_i\cup\cM_j$:
\[    (\cM_i\setminus\cM_j) \ind (\cM_j\setminus\cM_i) | \cM_i \cap \cM_j . \]
This gives a closed-form expression for the expected marginal frequencies in the $\cM_i\cup\cM_j$ in terms of the expected marginal frequencies in the $\cM_i$, so that~(\ref{normalscore2}) subject to~(\ref{freedom}) can be solved for $\bbeta$, using, for example, the Newton-Raphson method. 


\subsubsection{Remarks on the GEE method}

If the working covariance matrix is incorrect, the GEE method loses asymptotic efficiency compared to the asymptotically optimal ML method.
Above, we proposed a simple working covariance, which for univariate marginals corresponds to an independence working assumption. Touloumis et al.\ (2013) showed that this leads to a potentially big loss of efficiency if there is a strong dependence among the marginal observations. 
Efficiency can be improved by specifying a working covariance matrix that is closer to the truth, which can be done by specifying and estimating an appropriate parametric model for the marginals in~(\ref{margij}); 
Touloumis et al.\ obtained major improvements for univariate marginal models by modelling the bivariate marginals using homogeneous association models (see Agresti, 2013, Chapter 9, or Forcina and Kateri, 2021, for overviews of association models). 

The GEE method is a {\it quasi-likelihood} method. Another popular quasi-likelihood method is {\it composite likelihood}, which is based on a quasi-likelihood defined by multiplying certain marginal likelihoods; see, e.g., Molenberghs and Verbeke (2005, Chapter 9) for an overview. 
Composite likelihood has the advantage that it can be used both for marginal and conditional models. The GEE method has the advantage that, by improving the specification of the working covariance matrix, its asymptotic efficiency can be arbitrarily close to that of the ML method. 

Model comparison using GEE estimation is more difficult than using ML estimation. Model comparison and goodness-of-fit tools were developed by Rotnitzky and Jewell (1990), and the quasi-likelihood information criterion (QIC) developed by Pan (2001) is particularly popular.

\section{Areas of application}

\subsection{Directed graphical models} \label{graph}


Graphical models for categorical data associated with DAGs, or the more general chain graphs (Lauritzen, 1996), are marginal log-linear models in the sense of Bergsma and Rudas (2002). Parameterizations of these models have received considerable attention recently, see Rudas, Bergsma, and N\'emeth (2010), Marchetti and Lupparelli (2011), Evans and Richardson (2013), N\'emeth and Rudas (2013), and Nicolussi and Colombi (2017). For DAGs, the Markov property is
\begin{equation}\label{dag}
V_i \ind \nd(V_i) \mid \pa(V_i).
\end{equation}
Here, for every variable $V_i$, $\nd(V_i)$ denotes the non-descendants and $\pa(V_i)$ denotes the parents of $V_i$. The marginal log-linear parameterization of such models given in Rudas et al.\ (2010) is based on a well-numbering of the variables (Lauritzen et al., 1990), such that (\ref{dag}) is equivalent to
\begin{equation}\label{wellnumb}
V_i \ind  \pre(V_i)\setminus \pa(V_i) \mid \pa(V_i),
\end{equation}
where $\pre(V_i)$ is the set of variables preceding $V_i$ in the well-numbering. The parameterization proposed by
Rudas et al.\ (2010) is based on the marginals $\{V_i\} \cup \pre(V_i)$ which allows a parameterization as in Theorem~1.

Early work on statistical models associated with chain graphs includes Lauritzen and Wermuth (1989), Frydenberg (1990), Cox and Wermuth (1996), Andersson, Madigan, and Perlman (2001), Richardson (2003), Wermuth and Cox (2004), and Drton (2009).
For a component $\cK \subseteq \cV$ of a chain graph,  $\ND (\cK)$ is the set of nondescendants of $\cK$, i.e., the union of those
components, except $\cK$, for which no semi-directed path leads from any node in $\cK$ to any node in these components. $\PA(\cK)$ is the set of parents of $\cK$, i.e., the union of those components from which an arrow points to a node in $\cK $. The set of neighbours of $\cX \subseteq \cK$, $\nb(\mathcal{X})$, is the set of nodes in $\cK$ that are connected  to a node in $\mathcal{X}$
and $\pa(\mathcal{X})$ is the set of nodes from which an arrow points to any node in $\mathcal{X}$.

Chain graph models are defined by combinations of some of the following properties.

\vspace{1mm}

\begin{enumerate}
\item[P1] For all components $\cK$,
$
   \cK  \ind \left\{\ND(\cK ) \setminus \PA(\cK )\right\} \mid \PA(\cK),
$

\item[P2a] For all $\cK$ and $\mathcal{X} \subseteq \cK$,
$
\mathcal{X} \ind   \left\{ \cK \setminus \mathcal{X} \setminus
\nb(\mathcal{X}) \right\} \mid \left\{
 \PA(\cK) \cup \nb(\mathcal{X}) \right\},
$

\item[P2b] For all $\cK$ and $\mathcal{X} \subseteq \cK$,
$
 \mathcal{X} \ind   \left\{\cK  \setminus \mathcal{X} \setminus
\nb(\mathcal{X}) \right\} \mid
 \PA(\cK),
$

\item[P3a] For all  $\cK$ and  $\mathcal{X}
\subseteq \cK $,
$
\mathcal{X} \ind \left\{ \PA(\cK ) \setminus \pa(\mathcal{X}) \right\} \mid
\left\{ \pa(\mathcal{X}) \cup \nb(\mathcal{X}) \right\},
$

\item[P3b] For all  $\cK$ and  $\mathcal{X}
\subseteq \cK $,
$
\mathcal{X} \ind \left\{ \PA(\cK ) \setminus  \pa(\mathcal{X}) \right\}\mid
\pa(\mathcal{X}).
$
\end{enumerate}

\vspace{1mm}

The Type I Markov property (P1, P2a, P3a) is also called the Lauritzen--Wermuth--Frydenberg block-recursive
Markov property, see Lauritzen and Wermuth (1989) and Frydenberg (1990), and
the Type II Markov property (P1, P2a, P3b) is also called the Andersson--Madigan--Perlman block-recursive
Markov property, see Andersson et al.\ (2001).

Smoothness of Type I models is implied by the results of Frydenberg (1990) and is also easily obtained applying Theorem 1.

\begin{figure}[!]
  \vspace{-20mm}
    \hspace{6mm}
  \includegraphics[width=180mm]{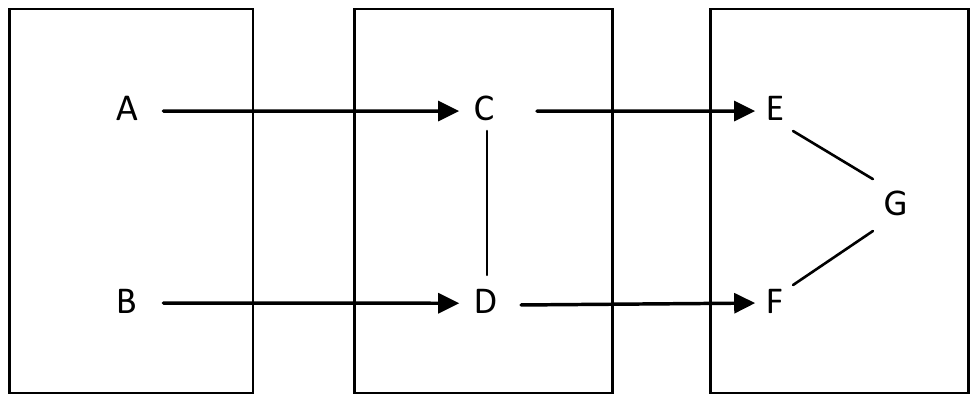}
\vspace{-195mm}
  \caption{Chain graph whose Andersson--Madigan--Perlman interpretation is a smooth model}
\label{ampgraph}
\end{figure}

The following example illustrates that marginal log-linear parameterizations may be used to establish smoothness of chain graph models belonging to model classes which also contain nonsmooth models.
The graph in Figure~\ref{ampgraph} with Type II  interpretation is a smooth  model and may be parameterized using the marginals  $AB$, $ABC$, $ABD$, $CDE$, $CDF$, $CDG$, $CDEG$, $CDFG$, $CDEFG$, $ ABCDEFG $. Type II models are not smooth in general, see Drton (2009), but in this case Theorem 1 implies  smoothness immediately.

Drton (2009) showed that Type IV models (P1, P2b, P3b) are smooth and gave a parameterization.  Lupparelli, Marchetti, and Bergsma (2009) illustrated through examples that these models are marginal log-linear.
We now apply the general method in Theorem~1 to obtain smoothness based on an interpretable parameterization, also implying  the number of degrees of freedom associated with a Type IV model.

\begin{theorem}
Assuming strictly positive discrete distributions, a  Type IV model for a chain graph is a hierarchical marginal log-linear model, and is, therefore, smooth. If the chain graph has components $\cK_1, \ldots , \cK_T$, that are well-numbered, the parameterization is based on the
marginals
\begin{equation}\label{margsth2}
\{ \PA(\cK_{t}) \cup \cX : \cX \subseteq \cK_{t} \}^*, \, \cK_1 \cup \ldots \cup \cK_{t} ,\,\,t=1, \ldots, T,
\end{equation}
where $\{ \,\,\,\}^*$ denotes a non-decreasing ordering of the elements of the set. The parameters set to zero to define the model are those associated with the effects in
\begin{equation}\label{zer}
\left\{ \dD(\cX, \cK_t \setminus \cX \setminus \nb(\cX),\PA(\cK_t)): \cX \subseteq \cK_t \right\} \cup
\end{equation}
$$
\left\{ \dD(\cX, \PA(\cK_t) \setminus \pa(\cX),\pa(\cX) ): \cX \subseteq \cK_t \right\} \cup
\dD(\cK_t, \PRE(\cK_t) \setminus \PA(\cK_t),\PA(\cK_t )),
$$
for all components $ \cK_t $, where $ \PRE(\cK_{t}) $ is the set of components that precede  $ \cK_{t} $.
\end{theorem}
The proof is given in Rudas et al.\ (2010). 
The parameters not set to zero, i.e., the ones not corresponding to~(\ref{zer}), parameterize the model. These parameters are associated with the same effects as those found by Marchetti and Lupparelli (2008) to have non-zero values in the examples they investigated, although the marginals used for the parameterization are different.




Further relevant work includes Marchetti and Lupparelli (2011), who described marginal log-linear parameterizations of chain graph models of the multivariate regression type. 
Evans and Richardson (2013) introduced a class of marginal models corresponding to Acyclic Directed Mixed Graphs (ADMGs), which contain both directed and bidirected edges. These models were shown to possess a smooth parameterization, and conditions were given for the parameterization to have a variation independence property. 
Nicolussi and Colombi (2017) considered  Type II chain graph models. This class of models is known to be not smooth, in general, but, by using a marginal log-linear parameterization, a smooth subclass could be identified.

\subsection{Path models}
\label{pathsection}

\begin{figure}
\caption{The graph of a path model}
\vspace{-10mm}
\hspace{30mm}
\includegraphics[trim=0 600 0 0, clip, width=110mm]{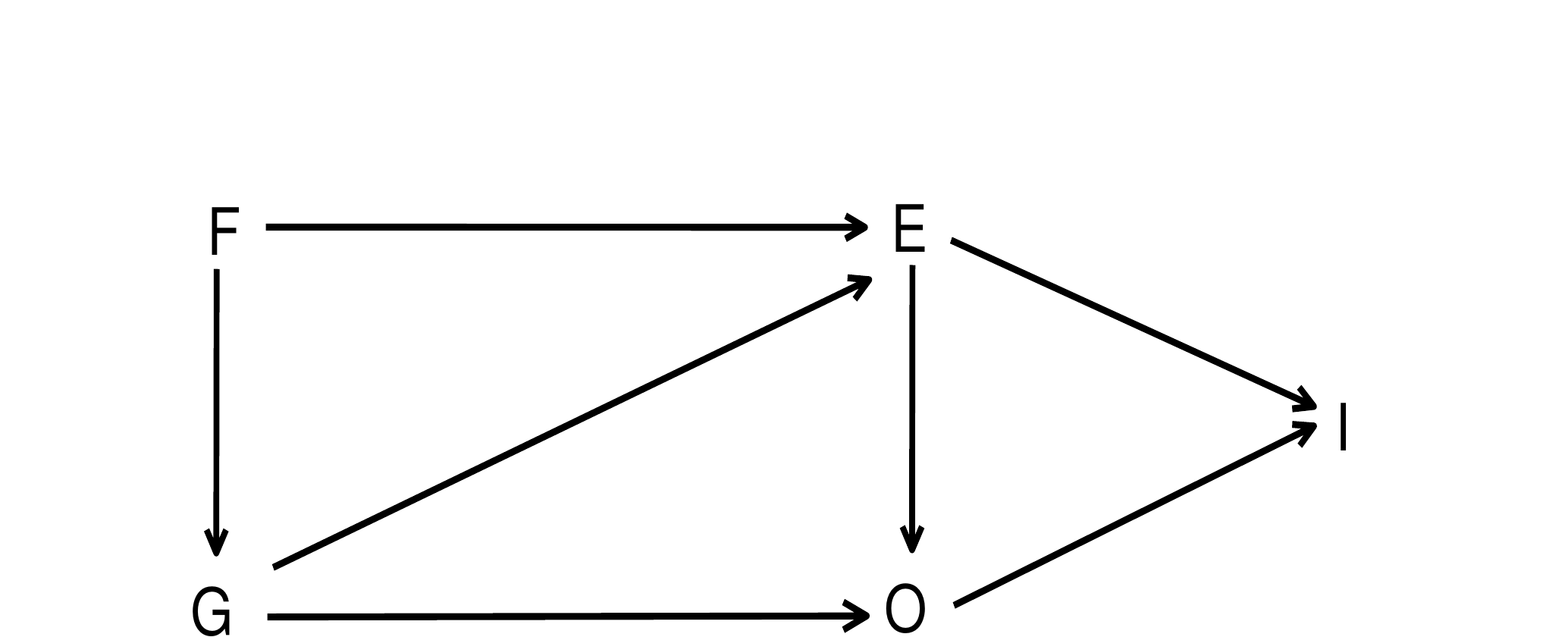}
\label{fig2}
\end{figure}

Path models have a long history in statistics and the basic idea is illustrated using Figure~\ref{fig2}. Intuitively, one may wish to use path models to describe a situation when variable $F$ influences $E$ and $G$, $G$ influences $E$ and $O$, $E$ influences $O$ and $I$, and $O$ influences $I$.\footnote{The notations for the variables are going to be clarified later.} In this case, however, one may say that in addition to the direct influence of $F$  on $E$, $F$ also has an indirect influence on $E$ through $G$. Similarly, $E$ influences  $I$ directly and indirectly. Also, one may wish to assume that only the influences  depicted in the graph exist among the variables. There is a further assumption, which is often made but usually remains implicit, namely that variables not taken into account only have negligible influences on those analysed. Path analysis aims at formulating these assumptions precisely and also at quantifying the magnitudes of the influences. 

To achieve these goals, motivated by Goodman (1973), Rudas et al. (2006) proposed the following $2$-step approach.

First, interpret the graph in Figure \ref{fig2} as a graphical Markov model and parameterize the distributions in this model as a marginal log-linear model.

The conditional independences associated with the graph are
$$
O \ind F | G E
$$
and
$$
I \ind F G | E O .
$$

These conditional independences may be conveniently imposed in a marginal log-linear model based on the marginals
$$
FGEO \makebox{ and } FGEOI,
$$
and are obtained, as implied by Theorem \ref{d-th}, by setting to zero the following marginal log-linear parameters
$$
\lambda^{FGEO}_{FO}, \, \lambda^{FGEO}_{FEO},  \, \lambda^{FGEO}_{FGO}, \, \lambda^{FGEO}_{FGEO},
$$
in the $FGEO$ marginal, and also
$$
\lambda^{FGEOI}_{FI}, \, \lambda^{FGEOI}_{GI}, \, \lambda^{FGEOI}_{FGI}, \, \lambda^{FGEOI}_{FEI}, \ \lambda^{FGEOI}_{FOI}, \, \lambda^{FGEOI}_{GEI}, 
$$
$$
\lambda^{FGEOI}_{GOI}, \lambda^{FGEOI}_{FGEI}, \,\lambda^{FGEOI}_{FGOI}, \,\lambda^{FGEOI}_{FEOI}, \, \lambda^{FGEOI}_{GEOI}, \, \lambda^{FGEOI}_{FGEOI}
$$
in the $FGEOI$ marginal, that is, in the whole table.

There is a total of $2^5=32$ parameters, and $16$ out of them are set to zero to imply the conditional independences. The remaining parameters parameterize the distributions in the  Markov model. The parameters in the present case are interpreted as measuring the strength of influence instead of association because of the inherent assumption behind using a directed graph to formulate the research hypothesis.

The remaining parameters belong to four disjoint groups, depending on the number of variables included in their effects:
the marginal log-linear parameter of the empty effect;
the marginal log-linear parameters with a single variable in their effects;
the marginal log-linear parameters with two variables in their effects; and
the marginal log-linear parameters with more than two variables in their effects. 

The parameters with more than $2$ variables in their effects quantify the joint influence of several variables on one variable, as these parameters are log-linear parameters (determined in a particular marginal) and possess the standard properties of log-linear parameters. For example, $\lambda^{FGEO}_{FGE}$, which is not set to zero, is a measure of the joint influence of $F$ and $G$ on $E$, in addition to  their individual influences.

Although the intention of the path model was to assume that such higher-order influences do not exist, their existence is not yet excluded. Indeed, it is very easy to find distributions for, say, three categorical variables, where there are no individual influences (all $2$-way marginal distributions are uniform), but two variables together completely determine the third one, illustrating that joint influences on top of the individual influences do exist (see, e.g., Rudas, 2018).

Therefore, in the next step of the path model definition, such higher-order interactions are excluded.

Second, assume that among the marginal log-linear parameters not set to zero in the first step, all those with more than two variables in their effect are equal to zero.

In the example, this implies setting to zero the following marginal log-linear parameters
$$
\lambda^{FGEO}_{FGE}, \, \lambda^{FGEO}_{GEO},
$$
and
$$
\lambda^{FGEOI}_{EOI}.
$$

To define a path model from  the graphical model, a further three parameters are set to zero. This means that the existence of the joint influence of $F$ and $G$ on $E$ and the joint influence of $G$ and $E$ on $O$ in the $FGEO$ marginal, and of the joint influence of $E$ and $O$ on $I$ in the $FGEOI$ marginal, are excluded.

The remaining $32-(16+3)=13$ marginal log-linear parameters parameterize all the distributions in the path model associated with the graph in Figure \ref{fig2}. These parameters are the univariate distributions of the variables and the strengths of the influences associated with the arrows in Figure \ref{fig2}.

The steps of the definition and  parameterization of the model are summarized in Table~\ref{pt}.

It has to be pointed out that these parameters provide a parameterization of all distributions in the path model. In a practical data analytic situation, this means that if a particular path model is used, then all relevant information from the data is summarized by the estimates of these parameters obtained from the data. N\'emeth  and   Rudas (2013a) provide such an example in the context of social status attainment with variables $F$ - father's education, $G$ - father's occupation, $E$ - son's education, $O$ - son's occupation, and $I$ - son's income. They found the path model associated with the graph in Figure \ref{fig2} well fitting to data for several countries, and gave estimates and interpretations of the parameters of the model. For further details of applications of marginal models to social mobility research, see   N\'emeth  and   Rudas (2013b).

\begin{table}
\[
\begin{array}{|c|c|c|}
\hline
\makebox{Marginals}  & \shortstack{ $FGEO$} & \shortstack{ $FGEOI$} \\ 
\hline
\shortstack{\makebox{Effects}} & \shortstack{$\emptyset$, {\it F}, {\it G}, {\it E}, {\it O}, {\it FG}, {\it FE}, {\it FO}, \\ {\it GE}, {\it  GO}, {\it EO}, {\it FGO},  {\it FGE},  \\ {\it FEO}, {\it GEO}, {\it FGEO}} & \shortstack{{\it I}, {\it FI}, {\it GI}, {\it EI}, {\it OI}, {\it FGI}, {\it FEI}, {\it FOI}, \\ {\it GEI},  {\it GOI}, {\it EOI}, {\it FGEI}, {\it FGOI}, \\ {\it FEOI}, {\it GEOI}, {\it FGEOI}} \\ 
\hline
\shortstack{\makebox{Effects set to zero to} \\  \makebox{define the graphical model}}  &\shortstack{ {\it FO}, {\it FEO}, {\it FGO}, {\it FGEO}} & \shortstack{ {\it FI}, {\it GI}, {\it FGI}, {\it FEI}, {\it FOI}, \\ {\it GEI}, {\it GOI}, {\it FGEI}, {\it FGOI},\\ {\it FEOI}, {\it GEOI}, {\it FGEOI}}   \\ 
\hline
\shortstack{\makebox{Effects set to zero to} \\  \makebox{define the path model}  } &\shortstack{ {\it FGE}, {\it GEO}} & \shortstack{{\it EOI}} \\ 
\hline
 \shortstack{\makebox{Remaining effects which} \\  \makebox{parameterize the path model}}  & \shortstack{$\emptyset$, $F$, {\it G}, {\it E}, {\it O}, {\it FG}, \\ {\it FE}, {\it GE}, {\it GO}, {\it EO} }& \shortstack{{\it I}, {\it EI}, {\it OI}} \\ 
\hline
\end{array}
\]
\caption{The definition and parameterization of the path model associated with Figure~\ref{fig2}}
\label{pt}
\end{table}

\subsection{Latent variable models}

When some relevant variables in an analysis cannot be observed, i.e., are latent, then the analysis of the observed variables applies to a marginal of the entire table. Therefore, latent variable models and marginal models are closely related. Under certain modelling assumptions, the joint distribution of the latent and observed variables may be estimated, but even in this case, testing of the model has to be restricted to a comparison of the estimated and observed marginal distributions.

For example, if the true position of someone on a left-right political scale, say $X$, is difficult to observe, then one may ask two related questions, say $A$ and $B$, which are indicators of $X$, but may not measure it precisely, rather with some measurement error. Thus, if  $A$ and $B$ are equal to $X$ perturbed by  measurement errors independent of $X$ and of each other, then one has
\begin{equation}\label{LV}
A \ind B | X.
\end{equation}
Measurement errors are usually assumed to be additive when the observations are numerical. For categorical data, measurement errors may also take different forms. For example, in the case of a binary variable the measurement error may change the category with a given probability. Then, the error is independent of the true category, if the probability of change does not depend on it. Or, for variables with multiple categories, the independent error may alter the category, so that reporting any category other than the true one has the same probability, which does not depend on the true category.

As $X$ is latent, and $A$ and $B$ are observed, (\ref{LV}) is a simple latent variable model. As seen above, it has many straightforward marginal log-linear model definitions. One can use the $X$ and $XAB$ marginals, or the $X$, $A$, $B$, $XAB$ marginals, but the definition may also be based on the $XAB$ whole table. In either case, the model is defined as
$$
\lambda^{XAB}_{AB} =\lambda^{XAB}_{XAB} = 0.
$$
To test the latent variable model (\ref{LV}), one has to rely on the observed data for the $AB$ marginal. The usual procedure is to specify the number of categories of the latent variable $X$ and to obtain estimates for the distribution $AB$ , subject to  (\ref{LV}) so that the likelihood of the observed data is maximized. To determine such estimates, usually the EM algorithm is applied (see, e.g., Rudas, 2018). Then, the estimates and the actual observations are compared using some statistical test.

Several more involved applications of the marginal modelling approach to latent variable models are described by Bergsma et al. (2009), Bergsma, Croon and Hagenaars (2013), and Hagenaars, Bergsma and Croon (2019). In one problem, there are two latent variables, $Y$ and $Z$, which are related. Their example refers to election forecasting and $Y$ is political party preference out of three parties and $Z$ is candidate preference, out of their respective candidates. These are seen as latent variables, which may only be observed in an imprecise way. The observed variables are the responses in two waves of a panel study to the party preference ($A$ and $B$) and to the candidate preference ($C$ and $D$) questions. This setup may be seen as an instance of the repeated measurement designs considered in Section \ref{repmeas}.  The model they fit is a graphical model of the path analysis type in the sense that the  highest order interactions allowed are $YZ, YA, YB, ZC, ZD$. 

To provide a marginal log-linear definition of this model, one may use the marginals $YZ, XA, XB, ZC, ZD, YZABCD$ and set all marginal log-linear parameters which are defined in the $YZABCD$ marginal equal to zero.

Estimates for the univariate marginal distributions of $Y$ and $Z$ show the relative popularities of the parties and of their respective candidates. Bergsma et al. (2009) investigate further the hypothesis that these marginal distributions are identical, i.e., the candidates are just as popular as the parties nominating them. This hypothesis is called latent marginal homogeneity.

To formulate latent marginal homogeneity, it is easiest to use the following marginals: $Y, Z, YZ, YA, YB, ZC, ZD, YZABCD$. Here also, the path model is imposed by setting to zero all parameters which are calculated in the $YZABCD$
marginal, and latent marginal homogeneity is imposed by requiring that
$$
\lambda^{Y}_{Y} = \lambda^{Z}_{Z},
$$
which is a marginal log-linear model.

Manifest variables are often considered indicators of the latent variables. The reliability of such an indicator is the extent to which the manifest variable is determined by the latent variable. This, of course, may be measured in many ways; one of these is based on the conditional distribution of the manifest variable, given the latent variable.

To consider a very simple model, let $Y$ and $Z$ be latent variables with manifest indicators $A$ and $
C$ respectively. We are not interested now in how the latent or the manifest variables are related, rather only in to what extent $Y$ determines $A$ and to what extent $Z$ determines $C$. By Theorem 2, if the $YZ$, $YA$ and $ZC$ marginals are used in a marginal log-linear parameterization, then
$\lambda^{YA}_A$ and $\lambda^{YA}_{YA}$ determine the conditional distribution of $A$ given $Y$, and $\lambda^{ZC}_C$ and $\lambda^{ZC}_{ZC}$  determine the conditional distribution of $C$ given $Z$.

If, now, all the variables are assumed to have identical categories, like party and candidate preference in the example above, then the requirement that
$$
\lambda^{YA}_A = \lambda^{ZC}_C \makebox{ and } \lambda^{YA}_{YA} = \lambda^{ZC}_{ZC}, 
$$
means equal reliability of the two manifest variables.

The strength of this approach to analysing reliability is that it can be combined with any other modelling assumption, given that the relevant marginals may be written in a non-decreasing order. For details and applications see Bergsma et al. (2009).

Marginal log-linear models with latent variables have also been considered in the context of capture-recapture models (Stanghellini and van der Heijden (2004), Bartoluci and Forcina (2006)). In this case also, observed variables are not necessarily independent conditionally on the latent variables.

\subsection{Further applications and extensions}

This section gives brief summaries of some further theoretical developments and interesting applications published in the literature.

Qaqish  and  Ivanova (2006) consider multivariate logistic parameterizations, which are generalized by the marginal log-linear parameterizations defined above, and provide results for the strong compatibility of such parameters.

Bartoluci  and  Forcina (2006) apply marginal log-linear parameterizations to develop models for the capture/recapture problem. For related work see also Turner (2007).

Forcina (2008) develops a marginal log-linear parameterization of latent class models with covariates and obtains identifiability results.

Bartolucci et al. (2012) develop a Bayesian approach to selecting the model best supported by the data from among a wide class of marginal models defined by equality or inequality constraints on generalized logits or generalized odds ratios. They use the Bayes factor to govern model selection.

Dardanoni et al. (2012) analyse intergenerational socioeconomic mobility tables for many countries, to test the monotonicity hypothesis stating that a higher socioeconomic class is never less advantageous than a lower one. They formulate this monotonicity as a marginal model, using the parameterization proposed by Bartolucci et al. (2007).

Shpitser et al. (2013) develop marginal log-linear parameterizations for nested Markov models and impose sparsity similar to the idea described in Section \ref{pathsection}.

Kuijpers et al. (2013a) propose methods to formulate and test hypotheses for the widely used measure of test score reliability, Cronbach's alpha, as marginal log-linear models.
\mbox{Kuijpers} et al. (2013b) provide standard errors of scalability coefficients in the case when the items are not binary, and also for large numbers of items, using a marginal modelling approach.

Colombi  and  Giordano (2015) parameterize the two components of a latent Markov model (the observed tie series and the unobserved Markov chain) with marginal log-linear parameters and show that relevant hypotheses may be formulated by setting some to zero.

Colombi  and  Forcina (2016) test inequality hypotheses for marginal log-linear parameters. They propose a likelihood-based procedure to test a set of equality constraints against positive departures from equality (the inequality constraints) and then the latter against the saturated model.  


Ntzoufras et al. (2019) discuss aspects of Bayesian inference for graphical marginal log-linear models. They provide a strategy to perform Markov chain Monte Carlo to obtain posterior densities. Their method also takes into account the requirement that the parameter values should be selected in a way which provides compatible marginal distributions.

Colombi et al. (2019) model the latent behaviour of raters with a binary variable indicating either one of two possible strategies. A marginal parameterization is used to link responses to underlying explanatory factors.

Bon et al. (2020) deal with disclosure limitation of sensitive or confidential data. They model the partial information provided by the data custodians as log-linear models on possibly overlapping marginals of a super-table and investigate methods of combining the available information. They also provide an application to Australian housing tenure transition data.

Nicolussi  and  Cazzaro (2020) analyse context specific independences, that is, independences which only hold in certain but not all category combinations of the variables involved, and show that hierarchical multinomial marginal models may be used to model such relationships.

Roverato, Lupparelli, and La Rocca (2013) and Lupparelli and Roverato (2017) consider {\it log-mean linear} parameterizations of marginal models for binary data. These are alternative marginal log-linear parameterizations to the ones considered in the present paper. Log-mean linear parameterizations have the interesting advantage of a closed-form likelihood.

Bergsma, Croon, and Hagenaars (2013) showed how marginal modelling methods can be extended to deal with complex sampling designs; in particular, they analysed a data set collected via a rotating panel design. The analysis there is carried out on data that are partially dependent. Furthermore, they showed how marginal modelling can be used for complex statistical models, giving an example of a data analysis using latent variables and both log-linear and non-log-linear constraints on the cell probabilities.









\end{document}